\documentclass{article}

    \usepackage{graphicx}
    \usepackage{caption}
    \DeclareCaptionFormat{nocaption}{}
    \usepackage{float}
    \usepackage{xcolor} % Allow colors to be defined
    \usepackage{enumerate} % Needed for markdown enumerations to work
    \usepackage{amsmath} % Equations
    \usepackage{amssymb} % Equations
    \usepackage{textcomp} % defines textquotesingle
    \AtBeginDocument{%
    }
    \usepackage{upquote} % Upright quotes for verbatim code
    \usepackage{eurosym} % defines \euro

    \usepackage{iftex}
    \ifPDFTeX
        \usepackage[T1]{fontenc}
        \IfFileExists{alphabeta.sty}{
              \usepackage{alphabeta}
          }{
              \usepackage[mathletters]{ucs}
              \usepackage[utf8x]{inputenc}
          }
    \else
        \usepackage{fontspec}
        \usepackage{unicode-math}
    \fi

    \usepackage{fancyvrb} % verbatim replacement that allows latex
    \usepackage{grffile} % extends the file name processing of package graphics
    \makeatletter % fix for old versions of grffile with XeLaTeX
    \@ifpackagelater{grffile}{2019/11/01}
    {
    }
    {
      \def\Gread@@xetex#1{%
        \IfFileExists{"\Gin@base".bb}%
        {\Gread@eps{\Gin@base.bb}}%
        {\Gread@@xetex@aux#1}%
      }
    }
    \makeatother
    \usepackage[Export]{adjustbox} % Used to constrain images to a maximum size
    \adjustboxset{max size={0.9\linewidth}{0.9\paperheight}}

    \usepackage{hyperref}
    \usepackage{titling}
    \usepackage{longtable} % longtable support required by pandoc >1.10
    \usepackage{booktabs}  % table support for pandoc > 1.12.2
    \usepackage{array}     % table support for pandoc >= 2.11.3
    \usepackage{calc}      % table minipage width calculation for pandoc >= 2.11.1
    \usepackage[inline]{enumitem} % IRkernel/repr support (it uses the enumerate* environment)
    \usepackage[normalem]{ulem} % ulem is needed to support strikethroughs (\sout)
    \usepackage{mathrsfs}
    
      \usepackage{arxiv}

    \definecolor{urlcolor}{rgb}{0,.145,.698}
    \definecolor{linkcolor}{rgb}{.71,0.21,0.01}
    \definecolor{citecolor}{rgb}{.12,.54,.11}

    \definecolor{ansi-black}{HTML}{3E424D}
    \definecolor{ansi-black-intense}{HTML}{282C36}
    \definecolor{ansi-red}{HTML}{E75C58}
    \definecolor{ansi-red-intense}{HTML}{B22B31}
    \definecolor{ansi-green}{HTML}{00A250}
    \definecolor{ansi-green-intense}{HTML}{007427}
    \definecolor{ansi-yellow}{HTML}{DDB62B}
    \definecolor{ansi-yellow-intense}{HTML}{B27D12}
    \definecolor{ansi-blue}{HTML}{208FFB}
    \definecolor{ansi-blue-intense}{HTML}{0065CA}
    \definecolor{ansi-magenta}{HTML}{D160C4}
    \definecolor{ansi-magenta-intense}{HTML}{A03196}
    \definecolor{ansi-cyan}{HTML}{60C6C8}
    \definecolor{ansi-cyan-intense}{HTML}{258F8F}
    \definecolor{ansi-white}{HTML}{C5C1B4}
    \definecolor{ansi-white-intense}{HTML}{A1A6B2}
    \definecolor{ansi-default-inverse-fg}{HTML}{FFFFFF}
    \definecolor{ansi-default-inverse-bg}{HTML}{000000}

    \definecolor{outerrorbackground}{HTML}{FFDFDF}

    \providecommand{\tightlist}{%
      \setlength{\itemsep}{0pt}\setlength{\parskip}{0pt}}
    \DefineVerbatimEnvironment{Highlighting}{Verbatim}{commandchars=\\\{\}}
    \let\Oldtex\TeX
    \let\Oldlatex\LaTeX
    \renewcommand{\TeX}{\textrm{\Oldtex}}
    \renewcommand{\LaTeX}{\textrm{\Oldlatex}}
\makeatletter
\def\PY@reset{\let\PY@it=\relax \let\PY@bf=\relax%
    \let\PY@ul=\relax \let\PY@tc=\relax%
    \let\PY@bc=\relax \let\PY@ff=\relax}
\def\PY@tok#1{\csname PY@tok@#1\endcsname}
\def\PY@toks#1+{\ifx\relax#1\empty\else%
    \PY@tok{#1}\expandafter\PY@toks\fi}
\def\PY@do#1{\PY@bc{\PY@tc{\PY@ul{%
    \PY@it{\PY@bf{\PY@ff{#1}}}}}}}
\def\PY#1#2{\PY@reset\PY@toks#1+\relax+\PY@do{#2}}

\@namedef{PY@tok@w}{\def\PY@tc##1{\textcolor[rgb]{0.73,0.73,0.73}{##1}}}
\@namedef{PY@tok@c}{\let\PY@it=\textit\def\PY@tc##1{\textcolor[rgb]{0.24,0.48,0.48}{##1}}}
\@namedef{PY@tok@cp}{\def\PY@tc##1{\textcolor[rgb]{0.61,0.40,0.00}{##1}}}
\@namedef{PY@tok@k}{\let\PY@bf=\textbf\def\PY@tc##1{\textcolor[rgb]{0.00,0.50,0.00}{##1}}}
\@namedef{PY@tok@kp}{\def\PY@tc##1{\textcolor[rgb]{0.00,0.50,0.00}{##1}}}
\@namedef{PY@tok@kt}{\def\PY@tc##1{\textcolor[rgb]{0.69,0.00,0.25}{##1}}}
\@namedef{PY@tok@o}{\def\PY@tc##1{\textcolor[rgb]{0.40,0.40,0.40}{##1}}}
\@namedef{PY@tok@ow}{\let\PY@bf=\textbf\def\PY@tc##1{\textcolor[rgb]{0.67,0.13,1.00}{##1}}}
\@namedef{PY@tok@nb}{\def\PY@tc##1{\textcolor[rgb]{0.00,0.50,0.00}{##1}}}
\@namedef{PY@tok@nf}{\def\PY@tc##1{\textcolor[rgb]{0.00,0.00,1.00}{##1}}}
\@namedef{PY@tok@nc}{\let\PY@bf=\textbf\def\PY@tc##1{\textcolor[rgb]{0.00,0.00,1.00}{##1}}}
\@namedef{PY@tok@nn}{\let\PY@bf=\textbf\def\PY@tc##1{\textcolor[rgb]{0.00,0.00,1.00}{##1}}}
\@namedef{PY@tok@ne}{\let\PY@bf=\textbf\def\PY@tc##1{\textcolor[rgb]{0.80,0.25,0.22}{##1}}}
\@namedef{PY@tok@nv}{\def\PY@tc##1{\textcolor[rgb]{0.10,0.09,0.49}{##1}}}
\@namedef{PY@tok@no}{\def\PY@tc##1{\textcolor[rgb]{0.53,0.00,0.00}{##1}}}
\@namedef{PY@tok@nl}{\def\PY@tc##1{\textcolor[rgb]{0.46,0.46,0.00}{##1}}}
\@namedef{PY@tok@ni}{\let\PY@bf=\textbf\def\PY@tc##1{\textcolor[rgb]{0.44,0.44,0.44}{##1}}}
\@namedef{PY@tok@na}{\def\PY@tc##1{\textcolor[rgb]{0.41,0.47,0.13}{##1}}}
\@namedef{PY@tok@nt}{\let\PY@bf=\textbf\def\PY@tc##1{\textcolor[rgb]{0.00,0.50,0.00}{##1}}}
\@namedef{PY@tok@nd}{\def\PY@tc##1{\textcolor[rgb]{0.67,0.13,1.00}{##1}}}
\@namedef{PY@tok@s}{\def\PY@tc##1{\textcolor[rgb]{0.73,0.13,0.13}{##1}}}
\@namedef{PY@tok@sd}{\let\PY@it=\textit\def\PY@tc##1{\textcolor[rgb]{0.73,0.13,0.13}{##1}}}
\@namedef{PY@tok@si}{\let\PY@bf=\textbf\def\PY@tc##1{\textcolor[rgb]{0.64,0.35,0.47}{##1}}}
\@namedef{PY@tok@se}{\let\PY@bf=\textbf\def\PY@tc##1{\textcolor[rgb]{0.67,0.36,0.12}{##1}}}
\@namedef{PY@tok@sr}{\def\PY@tc##1{\textcolor[rgb]{0.64,0.35,0.47}{##1}}}
\@namedef{PY@tok@ss}{\def\PY@tc##1{\textcolor[rgb]{0.10,0.09,0.49}{##1}}}
\@namedef{PY@tok@sx}{\def\PY@tc##1{\textcolor[rgb]{0.00,0.50,0.00}{##1}}}
\@namedef{PY@tok@m}{\def\PY@tc##1{\textcolor[rgb]{0.40,0.40,0.40}{##1}}}
\@namedef{PY@tok@gh}{\let\PY@bf=\textbf\def\PY@tc##1{\textcolor[rgb]{0.00,0.00,0.50}{##1}}}
\@namedef{PY@tok@gu}{\let\PY@bf=\textbf\def\PY@tc##1{\textcolor[rgb]{0.50,0.00,0.50}{##1}}}
\@namedef{PY@tok@gd}{\def\PY@tc##1{\textcolor[rgb]{0.63,0.00,0.00}{##1}}}
\@namedef{PY@tok@gi}{\def\PY@tc##1{\textcolor[rgb]{0.00,0.52,0.00}{##1}}}
\@namedef{PY@tok@gr}{\def\PY@tc##1{\textcolor[rgb]{0.89,0.00,0.00}{##1}}}
\@namedef{PY@tok@ge}{\let\PY@it=\textit}
\@namedef{PY@tok@gs}{\let\PY@bf=\textbf}
\@namedef{PY@tok@gp}{\let\PY@bf=\textbf\def\PY@tc##1{\textcolor[rgb]{0.00,0.00,0.50}{##1}}}
\@namedef{PY@tok@go}{\def\PY@tc##1{\textcolor[rgb]{0.44,0.44,0.44}{##1}}}
\@namedef{PY@tok@gt}{\def\PY@tc##1{\textcolor[rgb]{0.00,0.27,0.87}{##1}}}
\@namedef{PY@tok@err}{\def\PY@bc##1{{\setlength{\fboxsep}{\string -\fboxrule}\fcolorbox[rgb]{1.00,0.00,0.00}{1,1,1}{\strut ##1}}}}
\@namedef{PY@tok@kc}{\let\PY@bf=\textbf\def\PY@tc##1{\textcolor[rgb]{0.00,0.50,0.00}{##1}}}
\@namedef{PY@tok@kd}{\let\PY@bf=\textbf\def\PY@tc##1{\textcolor[rgb]{0.00,0.50,0.00}{##1}}}
\@namedef{PY@tok@kn}{\let\PY@bf=\textbf\def\PY@tc##1{\textcolor[rgb]{0.00,0.50,0.00}{##1}}}
\@namedef{PY@tok@kr}{\let\PY@bf=\textbf\def\PY@tc##1{\textcolor[rgb]{0.00,0.50,0.00}{##1}}}
\@namedef{PY@tok@bp}{\def\PY@tc##1{\textcolor[rgb]{0.00,0.50,0.00}{##1}}}
\@namedef{PY@tok@fm}{\def\PY@tc##1{\textcolor[rgb]{0.00,0.00,1.00}{##1}}}
\@namedef{PY@tok@vc}{\def\PY@tc##1{\textcolor[rgb]{0.10,0.09,0.49}{##1}}}
\@namedef{PY@tok@vg}{\def\PY@tc##1{\textcolor[rgb]{0.10,0.09,0.49}{##1}}}
\@namedef{PY@tok@vi}{\def\PY@tc##1{\textcolor[rgb]{0.10,0.09,0.49}{##1}}}
\@namedef{PY@tok@vm}{\def\PY@tc##1{\textcolor[rgb]{0.10,0.09,0.49}{##1}}}
\@namedef{PY@tok@sa}{\def\PY@tc##1{\textcolor[rgb]{0.73,0.13,0.13}{##1}}}
\@namedef{PY@tok@sb}{\def\PY@tc##1{\textcolor[rgb]{0.73,0.13,0.13}{##1}}}
\@namedef{PY@tok@sc}{\def\PY@tc##1{\textcolor[rgb]{0.73,0.13,0.13}{##1}}}
\@namedef{PY@tok@dl}{\def\PY@tc##1{\textcolor[rgb]{0.73,0.13,0.13}{##1}}}
\@namedef{PY@tok@s2}{\def\PY@tc##1{\textcolor[rgb]{0.73,0.13,0.13}{##1}}}
\@namedef{PY@tok@sh}{\def\PY@tc##1{\textcolor[rgb]{0.73,0.13,0.13}{##1}}}
\@namedef{PY@tok@s1}{\def\PY@tc##1{\textcolor[rgb]{0.73,0.13,0.13}{##1}}}
\@namedef{PY@tok@mb}{\def\PY@tc##1{\textcolor[rgb]{0.40,0.40,0.40}{##1}}}
\@namedef{PY@tok@mf}{\def\PY@tc##1{\textcolor[rgb]{0.40,0.40,0.40}{##1}}}
\@namedef{PY@tok@mh}{\def\PY@tc##1{\textcolor[rgb]{0.40,0.40,0.40}{##1}}}
\@namedef{PY@tok@mi}{\def\PY@tc##1{\textcolor[rgb]{0.40,0.40,0.40}{##1}}}
\@namedef{PY@tok@il}{\def\PY@tc##1{\textcolor[rgb]{0.40,0.40,0.40}{##1}}}
\@namedef{PY@tok@mo}{\def\PY@tc##1{\textcolor[rgb]{0.40,0.40,0.40}{##1}}}
\@namedef{PY@tok@ch}{\let\PY@it=\textit\def\PY@tc##1{\textcolor[rgb]{0.24,0.48,0.48}{##1}}}
\@namedef{PY@tok@cm}{\let\PY@it=\textit\def\PY@tc##1{\textcolor[rgb]{0.24,0.48,0.48}{##1}}}
\@namedef{PY@tok@cpf}{\let\PY@it=\textit\def\PY@tc##1{\textcolor[rgb]{0.24,0.48,0.48}{##1}}}
\@namedef{PY@tok@c1}{\let\PY@it=\textit\def\PY@tc##1{\textcolor[rgb]{0.24,0.48,0.48}{##1}}}
\@namedef{PY@tok@cs}{\let\PY@it=\textit\def\PY@tc##1{\textcolor[rgb]{0.24,0.48,0.48}{##1}}}

\makeatother

    \definecolor{incolor}{rgb}{0.0, 0.0, 0.5}
    \definecolor{outcolor}{rgb}{0.545, 0.0, 0.0}

    \hypersetup{
      breaklinks=true,  % so long urls are correctly broken across lines
      colorlinks=true,
      urlcolor=urlcolor,
      linkcolor=linkcolor,
      citecolor=citecolor,
      }
\begin{document}
    
    \maketitle

	\begin{abstract}

Central clearing counterparty houses (CCPs) play a fundamental role in mitigating the counterparty risk for exchange traded options. CCPs cover for possible losses during the liquidation of a defaulting member's portfolio by collecting initial margins from their members. In this article we analyze the current state of the art in the industry for computing initial margins for options, whose core component is generally based on a VaR or Expected Shortfall risk measure. We derive an approximation formula for the VaR at short horizons in a model-free setting. This innovating formula has promising features and behaves in a much more satisfactory way than the classical Filtered Historical Simulation-based VaR in our numerical experiments. In addition, we consider the neural-SDE model for normalized call prices proposed by [Cohen et al., arXiv:2202.07148, 2022] and obtain a quasi-explicit formula for the VaR and a closed formula for the short term VaR in this model, due to its conditional affine structure.

\end{abstract}

	\hypertarget{introduction}{%
\section{Introduction}\label{introduction}}

	The counterparty risk for exchange traded options is generally mitigated
thanks to Central Clearing Counterparty houses (CCPs), which take the
role of counterparty for option positions: the CCP becomes the seller in
front of the buyer and the buyer in front of the seller. In case of a
clearing member default, the CCP replaces the defaulting member until
its option positions are distributed among the surviving members through
a liquidation of the portfolio performed through brokers and/or through
an auction. The \(2008\) financial crisis entailed a strengthening of
the regulations for CCPs, requiring very robust risk frameworks in order
to achieve the task of covering potential losses incurred by a default
situation. As an example, the EMIR regulation lists the principles that
a CCP must adopt to safely operate. In particular, in order to cover for
the possible losses due to the liquidation of the defaulting member's
portfolio, the CCP requires from its members to deposit collateral in
form of initial margin, additional margins to cover liquidity and
concentration and/or specific risks, and default fund contribution.

We concentrate on the initial margin, which is supposed to cover for the
losses incurred in case of the liquidation of a given portfolio in
normal market conditions. Article 41 of \cite{EMIR} requires CCPs to
collect margins from the parties entering a transaction in a measure to
be sufficient to cover the CCP potential exposures while liquidating the
position. The margin must also be sufficient to cover at least
\(99.5\%\) of these exposures in the case of OTC derivatives, and
\(99\%\) for other financial products over the Margin Period Of Risk
(MPOR), as recommended in article 24 of \cite{RTS_EMIR}.

Since the drafting of EMIR regulation, CCPs have put in practice
different ways to compute margins for option portfolios. A first
notorious methodology for complex portfolios is the SPAN algorithm of
CME Group, which simulates joint risk scenarios for the underlier and
the implied volatility and infers a conservative margin from these
scenarios. However, this methodology has been overcome by more refined
ones, which in most cases apply a Filtered Historical Simulation
(FHS)-type algorithm \cite{barone1997var} to selected risk factors in
order to generate scenarios consistent with historical moves (examples
are the SPAN2 by CME and the IRM2 by ICE). FHS is widely used among CCPs
but its use on option markets is tricky and questionable. In particular,
a straightforward use of FHS breaks the structural relationships between
risk factors, possibly generating highly implausible scenarios.

Different techniques rather than FHS for options margining have been
studied in theory and eventually implemented, as the procyclicality
control model by Wong and Zhang from OCC \cite{wong2021procyclicality},
which relies on a dynamic scaling factor adjusting the dynamics of the
ATM IV log returns to be higher during low-volatility periods and lower
during high-volatility periods. More academic papers such as
\cite{cont2022simulation, cohen2021arbitrage} also look at the issue of
computing option initial margins, additionally ensuring the absence of
arbitrage for the generated scenarios. Indeed, in
\cite{cont2022simulation} the authors describe a generic algorithm which
penalizes arbitrageable scenarios (in a static sense) which can be
simply upgraded to any scenarios generation algorithm already in
production. In particular, the authors apply it to Generative
Adversarial Networks to simulate arbitrage-free implied volatility
surfaces. In \cite{cohen2021arbitrage}, an affine factor model for
normalized call option prices is firstly defined and then calibrated
minimizing dynamic and static arbitrages. Scenarios are subsequently
generated by neural networks which constrain the paths to live in the
polytope defined by the no static arbitrage conditions. In this
panorama, it is worth including the works by Bergeron at
al.~\cite{bergeron2022variational, ning2021arbitrage} on the Variational
Autoencoders used to reconstruct missing data on implied volatility
surface (eventually with no arbitrage), and which can be tweaked to
simulate scenarios based on historical movements.

In the present work we have two main objectives: the first one is to
provide a practical and concrete panorama in options margining; the
second, more ambitious, is to design a closed formula for the VaR of
option portfolios, which is easy to understand and to implement.
Specifically, we compute a short-term VaR formula which is completely
model-free and coincide with the exact one in the Stochastic Volatility
model and the particular affine factor model for normalized call prices
proposed in \cite{cohen2021arbitrage} as the neural-SDE model. For the
latter model, we show that it is actually possible to directly infer the
VaR formula without any need of simulating scenarios, so that once the
parameters of the model are calibrated, these can be plugged into a
quasi-explicit formula to obtain the required margin. Also, considering
the limit for small time steps, the formula becomes closed and it has
the same form of our short-term model-free formula. Testing the
short-term model-free formula, we obtain well-behaved margins which
actually beat the classic FHS ones in terms of regularity and adaptation
to the market current behavior. For these reasons, we believe that the
suggested short-term model-free formula could lay the foundations to a
practical model-free formula for options margining.

In the first part of this paper, we look at the mechanism of options'
initial margin adopted by CCPs in
\cref{the-mechanism-of-initial-margin-for-options}. In
\cref{initial-margin-for-options-in-the-industry-a-short-survey} we go
into detail in the practical implementations used by CCPs to calculate
initial margins, followed by an assessment of their pros and cons. In
the second part of the paper, we firstly describe the short-term
model-free formula in \cref{a-simple-short-term-model-free-formula} and
secondly derive the closed margin formula in the neural-SDE model for
normalized option prices in
\cref{quasi-explicit-formula-for-the-var-in-the-neural-sde-model}. We
conclude by performing numerical experiments in
\cref{numerical-experiments}.

	\begin{normalsize}
We thank Zeliade Systems and in particular Ismail Laachir for his detailed explanations on how the margining methodology works in CCPs, and Pierre Cohort and Niels Escarfail for being always available for helping with coding. A particular thank goes to Nicol\`o Filippas who, as an intern at Zeliade, has firstly studied and explained Cohen et al. papers to the team, hinting us to deepen our investigation of the model.
\end{normalsize}

We also thank Stefano De Marco who especially helped in the theoretical
adjustment of proofs and in the overall structuring of the article. All
remaining errors are ours.

	\hypertarget{the-mechanism-of-initial-margin-for-options}{%
\section{The mechanism of initial margin for
options}\label{the-mechanism-of-initial-margin-for-options}}

	CCPs charge clearing members, on a daily or intra-daily basis, with
total risk requirements that are computed from initial margins. The
initial margin aims at covering possible losses in the portfolio value
when liquidating it after a default, under normal market conditions, and
it is estimated considering a tail risk.

	Consider a portfolio, at time \(t\), with possibly both long \((L_i)_i\)
and short \((S_j)_j\) option positions (both \(L_i\) and \(S_j\) are
positive) with different strikes and expiries. In the case of default at
time \(t\), the CCP has to liquidate the portfolio in a Margin Period of
Risk (MPOR) of say \(h\) days (\(h\) is usually 2 days for
exchange-traded options). At date \(t+h\), the portfolio could have
undergone market movements, so that the CCP has to estimate its payoff
after liquidation.

The initial margin (IM) is then the Value-at-Risk (VaR) or
Expected-Shortfall (ES) at a confidence level of generally \(0.99\) of
the portfolio predicted P\&Ls:
\[\text{IM}(t) = -\text{VaR}_{0.99}\bigl(\text{P\&L}(t+h)\bigr)\] where
the minus sign ensures a positive margin value.

	At this point, the total risk requirement charged by the CCP does not
solely include the initial margin. Indeed, the CCP eventually adds to
the latter quantity some add-ons to take into account risks that are not
directly related to market moves. Among these, we typically find the
Wrong Way Risk add-on, the liquidity and concentration risk add-on and
possibly other specific add-ons:
\[\text{Total margin}(t) = \text{IM}(t) + \text{Add-ons}(t).\]

	Now, the total margin is floored by the Short Option Minimum (SOM). Deep
short OTM positions have very little risk since they will probably stay
OTM along the MPOR. However, their extreme risk is still not \(0\) and
the methodology should be able to capture it. This is generally not the
case for strikes very far from the ATM, because of the lack of
historical liquid data on these strikes. Then, to assure an enough
conservative margin, a secure floor should be applied to the risk
requirement. The SOM is generally the sum along all short positions of
the calibrated extreme costs for these options:
\[\text{Refined total margin}(t) = \max\bigl(\text{Total margin}(t);\ \text{SOM}(t)\bigr).\]

	At this point, the final total risk requirement is the refined total
margin adjusted by two other terms: minus the Net Option Value (NOV) on
equity-style options (options for which the premium is paid in full at
the settlement date, i.e.~one or two days after the option trade) and
the Unpaid Premium (UP), described below:
\[\text{Total risk requirement}(t) = \max\bigl(\text{Refined total margin}(t) - \text{NOV}(t) + \text{UP}(t);\ 0\bigr)\]
where the floor by zero is to avoid that the CCP pays and \begin{align*}
\text{NOV}(t) &= \sum_i L_i O_i(t) - \sum_j S_j O_j(t)\\
\text{UP}(t) &= \sum_{\text{$i$ unpaid}} L_i O_i(t) - \sum_{\text{$j$ undelivered}} S_j O_j(t)
\end{align*} with \(O\) denoting the option prices.

	For the NOV, let us consider for instance the case of a short option
position, where the option is of equity-style. In the case of a default,
the liquidation of this position would require to buy the option in the
market, which amounts to the CCP paying the option price at the time of
liquidation. This means that the initial margin for a short option
position should aim to cover the largest option price, up to a fixed
confidence level. On the other hand, the liquidation of a long option
position will always result in a positive inflow for the CCP, because
the CCP will sell the long option position and receive the option price.

The reason why the CCP applies the NOV can alternatively be explained
observing that the liquidation at the end of the MPOR, at time \(t+h\),
will result for the CCP in the monetary flow:
\[\text{Liquidation P\&L}(t+h) = \sum_i L_i O_i(t+h) - \sum_j S_j O_j(t+h).\]
The \(\text{Liquidation P\&L}\) can be expressed as the sum of the NOV
and the portfolio's value increment: \begin{align*}
\text{Liquidation P\&L}(t+h) =&\ \Bigl(\sum_i L_i O_i(t) - \sum_j S_j O_j(t)\Bigr) +\\
&+\Bigl(\sum_i L_i \bigl(O_i(t+h)-O_i(t)\bigr) - \sum_j S_j \bigl(O_j(t+h)-O_j(t)\bigr)\Bigr)\\
=&\ \text{NOV}(t) + \text{P\&L}(t+h).
\end{align*} Then, the CCP has to charge to the clearing member minus
the liquidation profit, i.e.~the predicted losses (the initial margin
appropriately adjusted by the add-ons and the SOM) minus the NOV on
equity-style options.

	The UP is charged by the CCP to cover from the risk of default of the
counterparts before the settlement date of the option premium, and it
corresponds to the net position of accrued option premiums which are
still unpaid (because the settlement date has not passed yet). In this
way, the difference between the UP and the NOV can be seen as a
(Contingent) Variation Margin for Options (VMO) not yet settled.

Consider a defaulting clearing member which is long an option before the
settlement date. The CCP will need to pay the premium to its counterpart
in the trade, and this will be done re-selling the option and collecting
its new premium, and using the VMO previously required to the buyer.
This latter component is needed to account for the difference between
the initially established option premium and today's one. Similarly, the
CCP has to liquidate defaulting short positions on not yet settled
options buying the option and delivering it to the buyer counterpart. To
do so, it will use the money from the buyer plus the VMO from the
defaulting seller.

	All in all, the final formula for the total risk requirement is
\[\text{Total risk requirement}(t) = \max\bigl(\max\bigl(\text{IM}(t) + \text{Add-ons}(t);\ \text{SOM}(t)\bigr) - \text{NOV}(t) + \text{UP}(t);\ 0\bigr).\]

In this article we will focus on the IM component of the total risk
requirement.

	\hypertarget{initial-margin-for-options-in-the-industry-a-short-survey}{%
\section{Initial margin for options in the industry: a short
survey}\label{initial-margin-for-options-in-the-industry-a-short-survey}}

	The total risk requirement mechanism and its different layers is
essentially the same across all CCPs, with possible differences in
wording. What really makes the difference among CCPs' requirements is
the way the IMs (and the add-ons) are computed. A notorious parametric
model for margining has been proposed by CME Group under the name of
SPAN.\footnote{\url{https://www.cmegroup.com/clearing/risk-management/span-overview.html}}
It consists in computing the \(\text{P\&L}\)s of the portfolio under
different risk scenarios depending on the combination of underlying
price changes and implied volatility changes. A similar model has been
implemented by ICE with the name of IRM. These models are particularly
tricky and overconservative, and for these reasons nowadays CCPs are
passing to new models. In particular, both SPAN and IRM models have been
upgraded to the corresponding
SPAN2\footnote{\url{https://www.cmegroup.com/clearing/risk-management/span-overview/span-2-methodology.html}}
and
IRM2\footnote{\url{https://www.theice.com/clearing/margin-models/irm-2/methodology}}
models, which both use the Filtered Historical Simulation (FHS)
techniques to create risk scenarios. Indeed, the majority of CCPs is now
adopting the FHS to compute IMs for option portfolios.

	\hypertarget{filtered-historical-simulation}{%
\subsection{Filtered Historical
Simulation}\label{filtered-historical-simulation}}

	The FHS has recently become the standard approach for VaR computations
among CCPs, especially on cash equity markets. The FHS technique is
indeed particularly efficient in cash equity and fixed income markets
for spot instruments, but it becomes more subtle in derivatives
clearing.

	The FHS model is particularly appreciated since it is essentially
data-driven and model-free, and it relies on few requirements to be
satisfied. For a given instrument to be cleared, firstly the CCP must
choose the risk factors which drive its price; let \(r_s\) denote their
returns, either logarithmic, absolute, or relative depending on the risk
factor. A key property that scaled returns must satisfy is stationarity
(see \cite{barone2001non}). Indeed, the FHS model relies on the
hypothesis that risk factors' returns at tomorrow's date \(t+1\) behave
as \[r_{t+1} = \eta\sigma_{t+1}\] where \(\sigma_{t+1}\) is the returns'
simulated conditional volatility at day \(t+1\) and \(\eta\) is drawn
from the historical observations \[\eta_s = \frac{r_s}{\sigma_s}.\] In
other words, the past historical return is \emph{re-contextualized} to
the current volatility context by the FHS devoling/revoling steps.

	Generally, the industry standard is to use an Exponentially Weighted
Moving Average (EWMA) variance estimator for the volatility. A EWMA
volatility with decay factor \(\lambda\) is computed as
\[\text{EWMA}_s = \sqrt{(1-\lambda)(r_s)^2 + \lambda\text{EWMA}_{s-1}^2},\]
with an eventual flooring in case of too low values. Then, the
historical volatility \(\sigma_s\) used to scale historical returns can
be calculated with two possible formulations: \(\sigma_s=\text{EWMA}_s\)
and \(\sigma_s=\text{EWMA}_{s-1}\) respectively. The two alternatives
are discussed in \cite{gurrola2015filtered}, section 7.1, where it is
acknowledged that they will lead to significantly different outcomes.

	When computing the IM for portfolios of options using the FHS
methodology, the CCP has to choose a set of risk factors, assessing the
stationarity property in particular. Together with the underlying value
(and possibly the interest/repo rate), also the Implied Volatility (IV)
has to be taken as a risk factor. Since the IV is actually a surface
which behaves differently depending on the strike and the maturity of
the option, two alternatives can be considered in order to generate IV
scenarios:

\begin{enumerate}
\def\labelenumi{\arabic{enumi}.}
\tightlist
\item
  Identify a fixed two-dimensional grid for the IV surface and define
  each point as a risk factor.
\item
  Choose a model for option prices and take its parameters as risk
  factors.
\end{enumerate}

For deeper insights on the VaR computation for options in a FHS approach
see \cite{gunnarsson2019filtered}.

	\hypertarget{implied-volatility-anchor-points}{%
\subsubsection{Implied Volatility anchor
points}\label{implied-volatility-anchor-points}}

	In the first alternative, the anchor points on the grid can be chosen
with fixed time-to-maturity or fixed rolling index as first coordinate,
and fixed log-forward moneyness, or fixed delta, or (equivalently) fixed
ratio between log-forward moneyness and square-root of time-to-maturity
as second coordinate. Since market data is more dense around the ATM
point for shortest maturities and spreads out for increasing maturities,
the fixed delta grid is generally preferred. Indeed, it implies a grid
in log-forward moneyness with a triangular shape, with a range that
starts from the ATM point and spreads out as the time-to-maturity
increases.

\begin{normalsize}
If choosing a dense grid guarantees a more precise fit of IVs in between the anchor points, it highly decreases computation performances and makes it difficult to identify general historical patterns in the IV surface dynamics. For this reason, Principal Component Analysis (PCA) can be performed in order to model the shifts of the surface. To cite an example, in \cite{yao2017managing}, the authors test the FHS method on the (PCA) principal components in a Karhunen-Lo\`eve decomposition and find that scenarios satisfy the conditions of no butterfly arbitrage (i.e. the requirement that for a fixed maturity, call prices must be non-increasing and convex with respect to the strike).
\end{normalsize}

Once scenarios are generated on the anchor points, the model still needs
an interpolation/extrapolation criterion to predict future prices on
points outside the grid. The criterion could be either a model for the
implied volatility (such as SVI) or for prices (such as SABR), which has
to be calibrated from the scenarios grid, or classic interpolations via
b-splines. The choice can be driven by arguments of non-arbitrability of
prices, or of best fit and computation efficiency of the algorithm.

	\hypertarget{implied-volatility-models}{%
\subsubsection{Implied Volatility
models}\label{implied-volatility-models}}

	In the second alternative, the CCP chooses a pricing model for options
and, once the stationarity property on the model parameters' returns is
verified, performs an FHS on the model parameters.

As an example, the SABR model is an industry standard and it is driven
by three parameters \(\alpha\), \(\beta\) and \(\rho\). Generally, the
\(\beta\) parameter is fixed a priori based on historical observations,
so that only the \(\alpha\) and \(\rho\) parameters need to be
estimated. After showing the stationarity of their returns in the target
market, the CCP can apply the FHS technique on the historical
observations of \(\alpha\) and \(\rho\), and use their drawn values to
simulate future prices.

Similarly, the Stochastic Volatility Inspired (SVI) model by Gatheral is
largely used among CCPs, and also among crypto funds, to model the
implied total variance. Its sub-model Slices SVI (SSVI) is sometimes
preferred since it has more tractable arbitrage-free requirements and
since it still fits data pretty well. SSVI has three parameters
\(\theta\), \(\varphi\) and \(\rho\) per each maturity, so that if the
stationarity of their returns is verified, the FHS technique can be
applied to obtain simulated prices. An example of this application can
be found in \cite{gunnarsson2019filtered}.

Lastly, a model which is sometimes considered is the Gaussian lognormal
mixture model as described by Glasserman and Pirjol in
\cite{glasserman2021w}. It consists in a convex combination of
Black-Scholes prices and the number of parameters depends on chosen
number of basis prices. Even though it is easy to implement, it
guarantees no arbitrage for slices and it has very good fitting ability,
the model is not easy to extend to full surfaces and it hides potential
issues when extrapolating in extreme events. Indeed, one it has the
theoretical property to have the same constant asymptotic level in the
two wings of the smile (Proposition 5.1 of \cite{glasserman2021w}), so
that while the calibration of market smiles could suggest a decreasing
shape, the calibrated smile with a lognormal Gaussian mixtures would
necessarily increase at large strikes, a pure model artefact. As a
consequence, while calibration fit could be good for liquid market data
(concentrated around the ATM point), in contexts such as the computation
of tail risks as in margins, the extrapolation at extreme strikes would
be misleading in those circumstances.

	\hypertarget{limitations-of-the-fhs}{%
\subsubsection{Limitations of the FHS}\label{limitations-of-the-fhs}}

	The FHS technique works well when a large history of risk factors is
stored, which is in fact a first immediate practical limitation. Indeed,
the possible number of scenarios for the FHS cannot be larger than the
available history, since the normalized returns are drawn from past
observations.

A second important drawback of the FHS methodology for complex products
is the capture of the joint dynamics of risk factors. Indeed, in the FHS
model, risk factors are re-scaled according to their own intrinsic
volatility, without any reference to other risk factors and this may
cause a discrepancy in the relationships between the risk factors, and
in particular in their correlations, as explained in section 7.2 of
\cite{gurrola2015filtered}. Furthermore, while the property of
stationarity of the normalized returns of single assets is generally
historically satisfied, this is more hardly the case for the returns of
IV points, resulting in more unstable and unnatural results for the FHS
methodology.

	Thirdly, FHS is relatively straightforward to implement, as far as the
risk factors under study do not have structural relationships which
could be destroyed by the core FHS algorithm. Unfortunately, this is
exactly the case for futures' curves and implied volatility surfaces.

Indeed, in the case of futures' curves, the FHS simulation considers a
set of fixed pillars (i.e.~futures' time-to-maturities) of the curves
today and apply the re-scaled corresponding past returns. For each
simulation, the simulated vector of futures values on the fixed pillars
should be consistent between the spot returns and the future returns.
However, this consistency is not guaranteed by the FHS simulations.

Similarly, when using the IV anchor points as risk factors, even if the
calibrated volatility surfaces are perfectly calibrated and
arbitrage-free, the volatility surfaces obtained by an FHS procedure
have no reason to be arbitrage-free in turn (and in general will not be,
because arbitrage-free surfaces do not have nice additive or
multiplicative properties). Furthermore, IV anchor points returns are
generally considered in absolute terms, which could cause negative
simulated implied volatilities. Flooring the latter quantities to \(0\)
is not a good choice, since: 1) prices for zero volatility are always
strictly lower than the market prices for European options; 2) since
call option prices are decreasing functions of the strike, a zero
volatility for an OTM call implies that all the calls with the same
maturity and larger strikes should also have a zero volatility, so that
also simulated implied volatility smiles should satisfy this property.

The possibility of generating scenarios such that each matrix of prices
indexed by the moneyness and time-to-maturity grid is arbitrage-free is
essentially an open question. A recent article \cite{cont2022simulation}
describes a weighted Monte-Carlo algorithm which penalizes arbitrageable
scenarios to obtain arbitrage-free simulations with higher probability.
We explain the model in \cref{arbitrage-free-simulations-for-options}.
Another alternative is to use parametric models of IV surfaces, for
which no-arbitrage conditions are available, and work at the level of
the parameters of such models. Yet, randomizing the model parameters may
produce a lot of instability. A noteworthy attempt is the neural-SDE
model of \cite{cohen2021arbitrage}, that we investigate further below,
which provides a consistent framework for this purpose.

	Finally, the FHS methodology is known to be procyclical as shown for
example in section 6 of \cite{gurrola2015filtered}. Procyclicality has
to be avoided because it implies margins which react too abruptly to
market changes, and this may cause liquidity issues to the clearing
members who have to post the corresponding collateral.

	\hypertarget{the-procyclicality-control-by-wong-and-zhang-options-clearing-corporation}{%
\subsection{The procyclicality control by Wong and Zhang (Options
Clearing
Corporation)}\label{the-procyclicality-control-by-wong-and-zhang-options-clearing-corporation}}

	Even though the FHS model is the most popular among CCPS, in the recent
years some new models for the options clearing are born and CCPs are
starting to look at these alternatives. An important feature in margin
requirements that CCPs should always try to mitigate is procyclicality
and we have seen that FHS does not properly satisfy this requirement.
The EMIR Regulatory Technical Standards of 2013 dedicates Article 28 to
the procyclicality control, detailing specific actions that CCPs have to
adopt for its limitation.

With this in mind, Wong and Zhang from the Options Clearing Corporation
(OCC) choose a model for options initial margin (see
\cite{wong2021procyclicality}) that guarantees to control procyclicality
thanks to a dynamic scaling factor that behaves as an inverted S-curve
and is higher during low-volatility periods and lower during
high-volatility ones.

	The model specifies the log-returns of the ATM IV at expiry \(T_j\) by
\begin{equation}\label{eqOCCModel}
\log\frac{\sigma_{t+h}(T_j,F_t(T_j))}{\sigma_t(T_j,F_t(T_j))} := \eta_t \Bigl(\frac{T_j}{T_1}\Bigr)^{-\alpha} \sqrt{h} z_t
\end{equation} where \(z_t\) is a normalized innovation, centered with
unit variance, \(F_t(T_j)\) is today's forward for maturity \(T_j\), and
\(T_1\) is the first quoted expiry. The factor \(\eta_t\) in turn is a
dynamic rescaling of the CBOE VVIXSM (VVIX), in particular
\[\eta_t = D(\sigma_t) \text{VVIX}_t\] where \(\sigma_t\) is the S\&P500
ATM IV of the short-term expiry (or any reference expiry, like the
one-month), and the scaling factor \(D(\sigma_t)\) is a sigmoid
function, which models a state transition from a risk point of view:
\[D(\sigma_t) = L+\frac{H}{1 + \exp{(\kappa(\sigma_t-\sigma^*))}}.\]
Here \(L\) is the minimum of the ratio between the long-term mean of the
historical vol-of-vol and the VVIX, \(H\) is the difference between the
maximum and the minimum of the latter ratio, \(\kappa\) is the growth
rate of the curve, and \(\sigma^*\) is the sigmoid's midpoint.

	The IV surface is recovered from the ATM IV dynamics considering the
second order approximation in \(\log\frac{K}{F_t(T)}\):
\[\sigma_t(T,K) \approx \sigma_t(T,F_t(T)) + \Sigma_t(T)\log\frac{K}{F_t(T)} + C_t(T)\Bigl(\log\frac{K}{F_t(T)}\Bigr)^2,\]
where \(\Sigma_t(T)\) and \(C_t(T)\) are respectively the ATM skew and
the ATM curvature.

In this way, knowing the distribution of \(z_t\) allows to perform
simulations of implied volatility surfaces and to compute an empirical
VaR.

Observe that the dynamics of the implied volatility in \cref{eqOCCModel}
are modeled for fixed strike and expiry, i.e.~for a fixed contract. This
differs with the majority of other models, whose dynamics are defined
for fixed time-to-maturity and log-forward moneyness.

	\hypertarget{arbitrage-free-simulations-for-options}{%
\subsection{Arbitrage-free simulations for
options}\label{arbitrage-free-simulations-for-options}}

	When computing an IM, the priority of the CCP is to be conservative
enough to cover for members' defaults, while not requiring too high
margins to keep its competitiveness in the market and avoiding
procyclicality. For this reason, arbitrage-free requirements are not
necessarily taken into account as seen for the FHS methodology. However,
simulating reliable scenarios (and so scenarios with no arbitrage)
allows to estimate more plausible margins, and avoids the pitfall of
paying for implausible scenarios.

	The article \cite{cont2022simulation} describes a cunning way to compute
an empirical VaR tweaking the simulations from any model in favor of
arbitrage-free simulations. The arbitrage considered in the article is
the static arbitrage, that, in case of options, can arise in both the
direction of time-to-maturity and the direction of log-forward
moneyness. Arbitrage-free call prices should:

\begin{enumerate}
\def\labelenumi{\arabic{enumi}.}
\tightlist
\item
  lie between the discounted intrinsic value (computed with respect to
  the forward) and the discounted forward;
\item
  increase in time-to-maturity at a fixed moneyness;
\item
  decrease in log-forward moneyness at a fixed time-to-maturity;
\item
  be a convex function of the log-forward moneyness.
\end{enumerate}

Note that in the article, the authors only address the last three
points, but the methodology can be easily extended to include the first
one.

Per each arbitrage situation, a penalization function is defined,
depending only on the normalized call prices surface on a fixed
time-to-maturity and log-forward moneyness discrete grid. Penalization
functions are null in case of no arbitrage and increase their value with
increasing arbitrageable grid points. The target arbitrage penalty
function is the sum of the three penalization functions, and it is null
if and only if the discrete call prices are free of arbitrage.

	At this point, the VaR calculation algorithm is straightforward:

\begin{enumerate}
\def\labelenumi{\arabic{enumi}.}
\tightlist
\item
  Simulate scenarios using the chosen initial model.
\item
  Per each simulated scenario:

  \begin{itemize}
  \tightlist
  \item
    evaluate the arbitrage penalty function;
  \item
    compute its weight inversely proportional to the arbitrage penalty
    function.
  \end{itemize}
\item
  Compute empirical VaR under the probability measure resulting from
  weights.
\end{enumerate}

Since weights prioritize arbitrage-free scenarios, the VaR calculation
hangs towards more reliable and possible values.

	The methodology holds for any model that simulates scenarios. It can
then be applied to both FHS and Monte-Carlo simulation models. In
particular, the authors apply it to a non-parametric generative model
for implied volatility surfaces called VolGAN (section 6 of
\cite{cont2022simulation}).

	\hypertarget{the-neural-sde-model}{%
\subsection{The neural-SDE model}\label{the-neural-sde-model}}

	In \cite{cohen2022estimating}, Cohen at al.~show very good empirical
results on options' VaR estimation. The results are based on a specific
model that the authors introduce in \cite{cohen2021arbitrage}, which
consists in a representation of normalized call prices via non-random
linear functions of some risk factors \(\xi_t\). The articles focus on
how to calibrate and consequently generate arbitrage-free call prices
surfaces via neural networks for the dynamics under consideration.

	In the neural-SDE model, the normalized call prices (i.e.~call prices
divided by the forward and discount factor) are affinely decomposed into
time-independent non-random surfaces \(G_i\) and time-dependent
stochastic combining factors \(\xi_{t,i}\in\mathbb R^d\):
\begin{equation}\label{eqCReisModel}
\begin{aligned}
c_t(\tau, k) &= G_0(\tau, k)+G(\tau,k)\cdot\xi_{t}\\
&= G_0(\tau, k)+\sum_{i=1}^{d}G_i(\tau,k)\xi_{t,i}
\end{aligned}
\end{equation} where \(\tau\) is the time-to-maturity and \(k\) is the
log-forward moneyness. The underlying asset \(S_t\) and the
time-dependent vector \(\xi_t\) evolve as \begin{equation}\label{eqSXi}
\begin{aligned}
dS_t &= \alpha(\xi_t)S_tdt+\beta(\xi_t)S_tdW_{0,t} & S_0=s_0\in\mathbb{R},\\
d\xi_t &= \mu(\xi_t)dt+\sigma(\xi_t)\cdot dW_t & \xi_0=\zeta_0\in\mathbb{R}^d,
\end{aligned}
\end{equation} where \(W_0\in\mathbb{R}\),
\(W=(W_1,\dots,W_d)^T\in\mathbb{R}^d\) are independent standard Brownian
motions under real-world measure \(\mathbb{P}\), and the hypothesis for
the existence and uniqueness of the processes hold,
i.e.~\(\alpha(\xi_t)\in L_\text{loc}^1(\mathbb{R})\),
\(\mu(\xi_t)\in L_\text{loc}^1(\mathbb{R}^d)\),
\(\beta(\xi_t)\in L_\text{loc}^2(\mathbb{R})\),
\(\sigma(\xi_t)\in L_\text{loc}^2(\mathbb{R}^{d\times d})\).

	Starting from these assumptions, the factors are decoded using different
PCA-based techniques to also ensure that the reconstructed prices are
more likely to be arbitrage-free both in a static and in a dynamic
sense. Absence of dynamical arbitrage is ensured through
Heath-Jarrow-Morton-type conditions while absence of static arbitrage is
ensured by imposing that each discretized normalized call prices'
surface respects a set of linear conditions \(Ac\geq b\) for some matrix
\(A\) and vector \(b\) (see \cite{cohen2020detecting}). Notice that
since the decomposition of the normalized call prices is affine and the
no static arbitrage conditions are linear, it is possible to rewrite the
latter conditions for \(\xi_t\) as
\(A\cdot G\cdot\xi_t\geq b - A\cdot G_0\).

	Given the history of market call prices, the factors \(G_i\) can be
calibrated for every grid point \((\tau_j,k_j)\) and factors
\(\xi_{s,i}\) for every past day \(s\) under the no arbitrage
constraints.

After the factors decoding, Cohen at al.~set up a supervised learning
process to estimate
\(\alpha(\xi_t), \beta(\xi_t), \mu(\xi_t), \sigma(\xi_t)\) via a maximal
likelihood function which ensures that the time series for \(\xi_t\)
evolves inside the convex polytope generated by the no static arbitrage
conditions.

	\hypertarget{empirical-var-in-the-neural-sde-model}{%
\subsubsection{Empirical VaR in the neural-SDE
model}\label{empirical-var-in-the-neural-sde-model}}

	Suppose we want to compute the VaR of a portfolio constituted of call
options at MPOR date \(t+h\), where \(h=n\delta t\) and \(\delta t\) is
the one day unit.

	Having the drift and diffusion functions for the time series of
\(\xi_t\) and the underlier from the model calibration, predictions can
be made with an Euler scheme in a Monte-Carlo fashion. In particular
(with an abuse of notation), processes values for the first step at
\(t+\delta t\) are \begin{equation*}
\begin{aligned}
S_{t+\delta t} &= S_t\exp\biggl(\Bigl(\alpha_t-\frac{\beta_t^2}2\Bigr) \delta t + \beta_t(W_{0,t+\delta t}-W_{0,t})\biggr),\\
\xi_{t+\delta t} &= \xi_t + \mu_t \delta t + \sigma_t(W_{t+\delta t}-W_t),
\end{aligned}
\end{equation*} where \begin{align*}
W_{0,t+\delta t}-W_{0,t} &= \sqrt{\delta t}X_0,\\
W_{t+\delta t}-W_t &= \sqrt{\delta t}X,
\end{align*} with independent standard Gaussian random variables
\(X_0\in\mathbb R\), \(X\in\mathbb R^d\).

In order to guarantee more stability of simulations, a tamed Euler
scheme can also be implemented.

The following steps are performed as above, using the latest values of
\(S\) and \(\xi\). At each step, new parameters \(\alpha\), \(\beta\),
\(\mu\) and \(\sigma\) can be estimated using the neural network
algorithm implemented in \cite{cohen2021arbitrage}.

Alternatively, assuming \(\alpha\), \(\beta\), \(\mu\) and \(\sigma\) to
be constant between \(t\) and \(t+h=t+n\delta t\), simulations for
\(S_{t+h}\) and \(\xi_{t+h}\) can be faster computed as \begin{equation}
\begin{aligned}
S_{t+h} &= S_t\exp\biggl(\Bigl(\alpha_t-\frac{\beta_t^2}2\Bigr) h + \beta_t(W_{0,t+h}-W_{0,t})\biggr),\\
\xi_{t+h} &= \xi_t + \mu_t h + \sigma_t(W_{t+h}-W_t).
\end{aligned}
\end{equation}

The predicted values of \(\xi_{t+h}\) can then be used to compute
predicted values of normalized call prices using \cref{eqCReisModel},
which can be de-normalized using the predicted values of \(S_{t+h}\).

The number of simulations that can be performed is arbitrary, so that a
stable value of the VaR can be computed as the empirical quantile of
simulated call prices.

	\hypertarget{limitations-of-the-neural-sde-model}{%
\subsubsection{Limitations of the neural-SDE
model}\label{limitations-of-the-neural-sde-model}}

	In this calibration routine of the neural-SDE model, there is an
important point which, according to us, should be taken into
consideration in applications: the \(G\) parameters are calibrated on
the history of market prices, but given their linear role in the
normalized call prices, there is little hope that a long history of call
prices will be well explained by the very same \(G\) factors. Indeed,
normalized call prices in this model are random linear combinations of
\emph{fixed} surfaces, so that one should probably expect the call
prices to maintain these fixed parameters for no more than a typical
period of one month or so, after which they should be re-calibrated. In
\cite{cohen2021arbitrage}, the \(G\) parameters are calibrated on a
\(17\)-years history, which might be far from being realistic in
practice. As a consequence, calibration fit is not as good as in other
more dynamic models. As an example, the mean absolute percentage error
(MAPE) computed by the authors in Table 2 of \cite{cohen2022estimating}
using two components of \(\xi_t\) is around \(4.61\%\) and \(5.40\%\),
while in our tests limiting the calibration window of \(G\) to one month
reduces the MAPE to about \(1.5\) percentage points. On the other hand,
it is not possible to simply calibrate the \(G\) parameters on shorter
periods of past history, since then the neural-SDE on the \(\xi_t\)
cannot be properly trained to estimate the model parameters, given the
too low amount of historical data.

To some extent, there is therefore a trade-off between the stationarity
of the model and its relevance - note though that one could argue that
this is a general situation for any model.

Furthermore, this stationarity of parameters is likely to be related to
the low procyclicality of the obtained VaR estimations that the authors
claim: because the \(G\) parameters are the same since several years,
the initial margin is indeed automatically less reactive to market
changes.

This being said, the neural-SDE model provides a consistent and
tractable framework which seems to us very promising.

	\hypertarget{the-market-data-in-input-of-the-margin-computation-and-marketmodel-add-ons}{%
\subsection{The market data in input of the margin computation, and
Market/Model
add-ons}\label{the-market-data-in-input-of-the-margin-computation-and-marketmodel-add-ons}}

	The models described above for margin computations (FHS, arbitrage-free
GANs, neural-SDE) have all in common the generation of scenarios for the
risk factors. In the case of options, these scenarios can only be
generated after an initial calibration of market prices using any
internal model, calibration which will then be reversed to get simulated
prices. Indeed, the CCP needs a model and/or an interpolation scheme to
get prices at any time-to-maturity and log-forward moneyness, and this
scheme is used since the beginning of the IM computation. As a result,
margins are based on model prices (i.e.~prices calibrated/interpolated
with the selected scheme), rather than market prices, and should then be
adjusted by a term taking into consideration how the initial discrepancy
between market and model prices propagates when computing the IM.

There can be two approaches to this issue in the context of a VaR-type
model:

\begin{enumerate}
\def\labelenumi{\arabic{enumi}.}
\tightlist
\item
  Apply the scenarios to the calibrated model prices, thus obtaining
  shocked model prices, and assume that the model P\&Ls are a good
  representative of the market P\&Ls, along each scenario. This means
  that the calibration error is assumed to be the same at the current
  date and at the future date along the shocked scenario.
\item
  Compute a \emph{Market/Model add-on}, which incorporates risk coming
  from the fact that the model which has been used to estimate the IM
  does not perfectly match market prices. Since the final risk
  requirement is computed on model prices and captures future movements
  of model prices, so that it could differ from the actual requirement
  needed for market prices, the Market/Model add-on estimates how large
  the difference between the market IM and the model IM is and adds it
  to the final requirement.
\end{enumerate}

	In the second approach, market P\&Ls can be decomposed in \(3\) terms:

\begin{itemize}
\tightlist
\item
  the difference between the portfolio price under the calibrated model
  and its market price: \(P^\text{mod}_t-P^\text{mkt}_t\);
\item
  the difference between the portfolio model prices along the scenario
  \(s\): \(\tilde P^\text{mod}_{t+h, s}-P^\text{mod}_t\);
\item
  the difference between the portfolio price under the calibrated model
  and its market price at the simulated date along the scenario \(s\):
  \(\tilde P^\text{mkt}_{t+h, s}-\tilde P^\text{mod}_{t+h, s}\).
\end{itemize}

The first of the \(3\) terms above is known and can be readily computed;
the second term is computed in the IM; the third term depends on each
scenario and upon the assumption on the distance between the market and
model prices at the future simulated date along each scenario. The
Market/Model add-on aims at covering this third source of risk.

	\hypertarget{a-simple-short-term-model-free-formula}{%
\section{A simple short-term model-free
formula}\label{a-simple-short-term-model-free-formula}}

	In this section we describe a new short-term model-free formula for
options VaR, which only depends on market data and does not need any
model-specific calibration. The short-term attribute depends on the fact
that approximations are performed in the MPOR component, so that the
shorter the MPOR, the more precise is the formula.

	In the following we will denote with \(\text{DF}_t(\tau)\) and
\(F_t(\tau)\) the discount factor and the forward value for
time-to-maturity \(\tau\) evaluated at time \(t\). We work under the
hypothesis of known constant rates between today date \(t\) and the MPOR
date \(t+h\), so that for a given time-to-maturity, discount factors are
constant and forward values are proportional to the underlier \(S_t\).
In particular we write \(F_t(\tau) = f(\tau)S_t\). We call \(\delta t\)
the one day unit and consider an MPOR \(h=n\delta t\) of \(n\) days.
Finally, we denote respectively by \(p_Y\) and \(F_Y\) the probability
density function and the cumulative density function of a generic random
variable \(Y\). The cumulative density function and the probability
density function of a standard Gaussian random variable are denoted with
\(\Phi\) and \(\varphi\) respectively. Also, when considering the
distribution of the underlier \(S_{t+h}\) at time \(t+h\), we actually
mean the distribution conditional to quantities at time \(t\)
(i.e.~\(S_t\) and other risk factors \(\xi_t\)).

	In the following sections we will always consider a portfolio of Vanilla
calls with price at time \(t\) given by \begin{equation*}
\Pi_t(S_t,\xi_t) = \sum_i\pi_iC\Bigl(T_i-t,\log\frac{K_i}{f(T_i-t)S_t};S_t,\xi_t\Bigr)
\end{equation*} where \(C(\tau,k;S_t,\xi_t)\) is a generic model price
of a call with time-to-maturity \(\tau\) and log-forward moneyness
\(k\), depending on the current value of the underlier \(S_t\) and of
the other possible risk factors \(\xi_t\) (as for example the implied
volatility in the short-term model-free case). The P\&Ls are defined as
the finite differences of the portfolio over the MPOR:
\[\text{P\&L}:=\Pi_{t+h}(S_{t+h},\xi_{t+h})-\Pi_t(S_t,\xi_t).\]

The \(h\)-days VaR at confidence level \(\theta\) (close to \(1\)) of
the portfolio is the quantity \(v(\theta,h)\) such that
\[P\bigl(\text{P\&L}\leq v(\theta,h)\bigr) = 1-\theta.\] Sometimes we
will need to develop the above expression using conditional
probabilities. In particular, it holds
\begin{equation}\label{eqIntPPnLCondS}
\begin{aligned}
P\bigl(\text{P\&L}\leq v(\theta,h)\bigr) &= E[\mathbbm{1}_{\text{P\&L}\leq v(\theta,h)}]\\
&= E\bigl[E[\mathbbm{1}_{\text{P\&L}\leq v(\theta,h)}| S_{t+h}]\bigr]\\
&= \int_0^\infty P\bigl(\text{P\&L}\leq v(\theta,h)|S_{t+h}=s\bigr)\,dF_{S_{t+h}}(s).
\end{aligned}
\end{equation} In the case of existence of a probability function for
\(S_{t+h}\), the latter expression can be written as
\[P\bigl(\text{P\&L}\leq v(\theta,h)\bigr) = \int_0^\infty p_{S_{t+h}}(s)P\bigl(\text{P\&L}\leq v(\theta,h)|S_{t+h}=s\bigr)\,ds.\]

	\hypertarget{the-black-scholes-case-and-the-stochastic-volatility-case}{%
\subsection{The Black-Scholes case and the Stochastic Volatility
case}\label{the-black-scholes-case-and-the-stochastic-volatility-case}}

	Before introducing the short-term model-free VaR formula, we firstly
look at some prototypical examples such as the Black-Scholes and the
Stochastic Volatility cases.

	In the classic Black-Scholes case, the underlier is a geometric Brownian
motion \[dS_t = \alpha_tS_tdt + \beta_tS_tdW_t\] under the real-world
probability measure. Applying Ito's lemma, portfolio prices are
processes such that
\[d\Pi_t(S_t) = \biggl(S_t\alpha_t\frac{d}{dS_t}\Pi_t(S_t) + \frac{1}{2}S_t^2\beta_t^2\frac{d^2}{dS_t^2}\Pi_t(S_t)\biggr)dt + S_t\beta_t\frac{d}{dS_t}\Pi_t(S_t)dW_t.\]
Writing \(dW_t\) as \(\sqrt h X\) where \(X\) is a standard Gaussian
random variable and approximating the above expression at the first
order we have that the \(\text{P\&L}\)s have the form
\[\Pi_{t+h}(S_{t+h})-\Pi_t(S_t)\approx S_t\beta_t\frac{d}{dS_t}\Pi_t(S_t)\sqrt{h}X.\]
Then, it is easy to compute the VaR with a first order approximation:
\[P\bigl(\text{P\&L}\leq v(\theta,h)\bigr) = P\biggl(\frac{d}{dS_t}\Pi_t(S_t)X\leq \frac{v(\theta,h)}{S_t\beta_t\sqrt h}\biggr)\]
so that
\[v(\theta,h) = \Phi^{-1}(1-\theta)S_t\beta_t\Bigl|\frac{d}{dS_t}\Pi_t(S_t)\Bigr|\sqrt h.\]

	The above reasoning can actually be generalized to Stochastic Volatility
models where the volatility of the underlier is a stochastic process
with volatility \(\sigma_t\): \begin{align*}
dS_t &= \alpha_tS_tdt + \xi_tS_tdW_{0,t}\\
d\xi_t &= \mu_tdt + \sigma_tdW_t\\
dW_{0,t}dW_t &= \rho_tdt.
\end{align*} In the above formulation we have dropped the dependency of
volatility parameters in the volatility itself,
i.e.~\(\mu_t=\mu_t(\xi_t)\) and \(\sigma_t=\sigma_t(\xi_t)\). Indeed, in
order to guarantee the positivity of the volatility there must be such a
dependency. Applying Ito's lemma to the portfolio \(\Pi_t(S_t,\xi_t)\)
of option prices generated by the pricing version of the Stochastic
Volatility model, one finds
\[d\Pi_t(S_t,\xi_t) = a_tdt + S_t\xi_t\frac{d}{dS_t}\Pi_t(S_t,\xi_t)dW_{0,t} + \sigma_t\frac{d}{d\xi_t}\Pi_t(S_t,\xi_t)dW_t\]
where \begin{align*}
a_t =&\ \alpha_tS_t\frac{d}{dS_t}\Pi_t(S_t,\xi_t) + \mu_t\frac{d}{d\xi_t}\Pi_t(S_t,\xi_t) + \frac{\xi_t^2S_t^2}{2}\frac{d^2}{dS_t^2}\Pi_t(S_t,\xi_t) +\\
&+ \frac{\sigma_t^2}{2}\frac{d^2}{d\xi_t^2}\Pi_t(S_t,\xi_t) + \xi_tS_t\sigma_t\rho_t\frac{d^2}{dS_td\xi_t}\Pi_t(S_t,\xi_t).
\end{align*} Considering the finite increments of the portfolio and
neglecting linear terms for \(h\) going to \(0\), the form of the
\(\text{P\&L}\)s becomes
\[\Pi_{t+h}(S_{t+h},\xi_{t+h})-\Pi_t(S_t,\xi_t)\approx S_t\xi_t\frac{d}{dS_t}\Pi_t(S_t,\xi_t)\sqrt{h}X_0 + \sigma_t\frac{d}{d\xi_t}\Pi_t(S_t,\xi_t)\sqrt{h}X,\]
where \(X_0\) and \(X\) are standard jointly normal random variables
with correlation \(\rho_t\) equal to the correlation of the two Brownian
motions. Then, any combination of \(X_0\) and \(X\) is still normal and
the VaR of the portfolio is \begin{equation}\label{eqGenericVaR}
v(\theta,h) = \Phi^{-1}(1-\theta)\sqrt{\Bigl(S_t\xi_t\frac{d}{dS_t}\Pi_t(S_t,\xi_t)\Bigr)^2 + \Bigl(\sigma_t\frac{d}{d\xi_t}\Pi_t(S_t,\xi_t)\Bigr)^2 + 2\rho_tS_t\xi_t\sigma_t\frac{d}{dS_t}\Pi_t(S_t,\xi_t)\frac{d}{d\xi_t}\Pi_t(S_t,\xi_t)}\sqrt h.
\end{equation}

	\hypertarget{a-short-term-model-free-formula}{%
\subsection{A short-term model-free
formula}\label{a-short-term-model-free-formula}}

	Driven by the results in the Black-Scholes and the Stochastic Voaltility
case, we generalize the VaR formulas to a short-term model-free formula
which can be applied to any historical series of spot and option prices.

	With this aim, we rather work using the implied volatility, which can
always be computed from market prices using a root-finding algorithm
applied to the classic Black-Scholes pricing formula
\[\text{BS}_t\bigl(k, \tau, \omega, F_t(\tau), \text{DF}_t(\tau), \sigma^\text{imp}_t(k,\tau)\bigr) = \omega\text{DF}_t(\tau)F_t(\tau)\bigl(\Phi(\omega d_1) - e^k\Phi(\omega d_2)\bigr)\]
where \(k=\log\frac{K}{F_t(\tau)}\) is the log-forward moneyness,
\(\tau=T-t\) is the time-to-maturity,
\[d_{1,2} = -\frac{k}{\sigma_t^\text{imp}(k,\tau)\sqrt\tau} \pm \frac{\sigma_t^\text{imp}(k,\tau)\sqrt\tau}2\]
and \(\omega=+1\) if the option is a call, \(-1\) if it is a put.

	When computing risks, the implied volatility
\(\sigma^\text{imp}_t(k,\tau)\) is generally considered as a risk factor
together with the underlier \(S_t\). For this reason, we write it as a
function of a driving factor \(\xi_t\):
\(\sigma^\text{imp}_t(k,\tau) = \sigma(k,\tau,\xi_t)\), so that the
dynamics of the two risk factors are \begin{equation}
\begin{aligned}
dS_t &= \alpha_tS_tdt + \beta_tS_tdW_{0,t}\\
d\xi_t &= \mu_tdt + \eta_tdW_t\\
dW_{0,t}dW_t &= \rho_tdt.
\end{aligned}\tag{H1}
\end{equation}

Observe that the implied volatility risk factor depends on the
log-forward moneyness and the time-to-maturity rather than the contract
strike and its maturity. Indeed, the time series of a fixed contract is
available since its issue date and is then limited in time. Furthermore,
we do not expect its implied volatility to have any nice statistical
property of stationarity that could legitimate drawing meaningful
forecasts for its historical returns between time \(t\) and \(t+h\). On
the contrary, we expect that the market encode the implied volatility
risk rather in a log-forward moneyness, time-to-maturity map, meaning
that the time series of the implied volatility at a fixed point in this
relative grid will have much nicer features.

	Let us consider a portfolio of calls and puts written as Black-Scholes
functions: \begin{equation}\label{eqPiBS}
\begin{aligned}
\Pi_t(S_t,\xi_t) &= \sum_i\pi_i\text{BS}_t\biggl(\log\frac{K_i}{f(T_i-t)S_t}, T_i-t, \omega_i, F_t(T_i-t), \text{DF}_t(T_i-t), \sigma\Bigl(\log\frac{K_i}{f(T_i-t)S_t},T_i-t,\xi_t\Bigr)\biggr)\\
&=: \sum_i\pi_i\text{BS}_t^i.
\end{aligned}
\end{equation} Repeating the steps in
\cref{the-black-scholes-case-and-the-stochastic-volatility-case}, the
approximated formula for the VaR becomes
\[v(\theta,h) = \Phi^{-1}(1-\theta)\sqrt{\Bigl(S_t\beta_t\frac{d}{dS_t}\Pi_t(S_t,\xi_t)\Bigr)^2 + \Bigl(\eta_t\frac{d}{d\xi_t}\Pi_t(S_t,\xi_t)\Bigr)^2 + 2\rho_tS_t\beta_t\eta_t\frac{d}{dS_t}\Pi_t(S_t,\xi_t)\frac{d}{d\xi_t}\Pi_t(S_t,\xi_t)}\sqrt h.\]
This formula is far from being model-free for two reasons:

\begin{itemize}
\tightlist
\item
  The term \(\frac{d}{dS_t}\Pi_t(S_t,\xi_t)\) is the full derivative of
  the portfolio \(\Pi_t\) with respect to \(S_t\), which also involves
  the derivative of prices with respect to the implied volatility, since
  it depends on \(k=\log\frac{K}{f(\tau)S_t}\). As a consequence, it
  does not correspond to the Black-Scholes delta and its expression must
  be made explicit.
\item
  The term \(\frac{d}{d\xi_t}\Pi_t(S_t,\xi_t)\) is the derivative of the
  portfolio with respect to \(\xi_t\), and it does not coincide with
  what the market indicates with vega, i.e.~the sensibility of the
  portfolio to the option volatility.
\end{itemize}

Given the above, we shall rather develop the dynamics of the portfolio
as a function of \(S_t\) and \(\sigma(k,\tau,\xi_t)\) where \(k\) also
depends on \(S_t\).

	Observe that for fixed \(k\) and \(\tau\), we have
\begin{equation}\label{eqSigmaImplFalseDiff}
\begin{aligned}
d\sigma_t &= \partial_\xi\sigma_t\,d\xi_t\\
&= \mu_t\partial_\xi\sigma_t\,dt + \eta_t\partial_\xi\sigma_t\,dW_t
\end{aligned}
\end{equation} where \(\sigma_t=\sigma(k,\tau,\xi_t)\). We define
\(\zeta_t(k,\tau,\xi_t):=\eta_t\partial_\xi\sigma(k,\tau,\xi_t)\).

On the other hand, writing \(k=\log\frac{K}{f(\tau)S_t}\) and
\(\tau=T-t\), we rather find \begin{equation}\label{eqSigmaImplDiff}
\begin{aligned}
d\sigma_t &= a_t\,dt - \frac{\partial_k\sigma_t}{S_t}\,dS_t + \partial_\xi\sigma_t\,d\xi_t\\
&= \bigl(a_t - \alpha_t\partial_k\sigma_t + \mu_t\partial_\xi\sigma_t\bigr)\,dt - \beta_t\partial_k\sigma_t\,dW_{0,t} + \zeta_t\,dW_t
\end{aligned}
\end{equation} where
\(a_t=-\partial_\tau\sigma_t+\partial_k\sigma_t\frac{\partial_\tau f}{f}+\frac{\beta_t^2}2(\partial^2_k\sigma_t+\partial_k\sigma_t)+\frac{\eta_t^2}2\partial^2_\xi\sigma_t-\rho_t\beta_t\eta_t\partial_\xi\partial_k\sigma_t\)
and \(\sigma_t = \sigma\bigl(\log\frac{K}{f(T-t)S_t},T-t,\xi_t\bigr)\).

Using \cref{eqSigmaImplDiff} and ignoring the terms in \(h\) in the
finite scheme of the portfolio increments, the \(\text{P\&L}\)s assume
the form
\[\Pi_{t+h}(S_{t+h},\xi_{t+h})-\Pi_t(S_t,\xi_t) \approx \beta_t\bigl(S_t\sum_i\pi_i\partial_S\text{BS}_t^i - \sum_i\pi_i\partial_k\sigma_t^i\partial_\sigma\text{BS}_t^i\bigr)\sqrt h X_0 + \sum_i\pi_i\zeta_t^i\partial_\sigma\text{BS}_t^i\sqrt h X\]
where
\(\sigma_t^i = \sigma\bigl(\log\frac{K_i}{f(T_i-t)S_t},T_i-t,\xi_t\bigr)\),
\(\zeta_t^i = \zeta_t\bigl(\log\frac{K_i}{f(T_i-t)S_t},T_i-t,\xi_t\bigr)\),
and \(X_0\) and \(X\) are standard jointly normal random variables with
correlation \(\rho_t\).

	Let us denote \(\Delta_t^i:=\partial_S\text{BS}_t^i\) and
\(\mathcal V_t^i:=\partial_\sigma\text{BS}_t^i\). Using the same proof
as in the Stochastic Volatility case of
\cref{the-black-scholes-case-and-the-stochastic-volatility-case}, we
finally find a VaR formula on an MPOR horizon of \(h\) days of the form:
\begin{equation}\label{eqModelFreeVaR}
\begin{aligned}
\text{VaR}_{\theta,t}(h) &= \Phi^{-1}(1-\theta)\sqrt{c_t^2 + q_t^2 + 2\rho_tc_tq_t}\sqrt h\\
c_t &= \beta_t\bigl(S_t\sum_i\pi_i\Delta_t^i - \sum_i\pi_i\mathcal V_t^i\partial_k\sigma_t^i\bigr)\\
q_t &= \sum_i\pi_i\zeta_t^i\mathcal V_t^i.
\end{aligned}\tag{Short-term formula}
\end{equation}

	This expression is actually model-free. Indeed, the terms \(\Delta_t^i\)
and \(\mathcal V_t^i\) are respectively the Black-Scholes delta and vega
of the \(i\)-th option in the portfolio. In particular they correspond
to \begin{align*}
\Delta_t\bigl(k, \tau, \omega, \sigma_t^\text{imp}(k,\tau)\bigr) &= \omega\Phi(\omega d_1),\\
\mathcal V_t\bigl(k, \tau, F_t(\tau), \text{DF}_t(\tau), \sigma_t^\text{imp}(k,\tau)\bigr) &= \text{DF}_t(\tau)F_t(\tau)\varphi(d_1)\sqrt\tau.
\end{align*}

The volatility \(\beta_t\) of the underlying spot \(S_t\) can be
computed looking at historical moves. For example, it could be a EWMA
volatility on log-returns appropriately rescaled by the square-root of
the returns' distance \(h_r\) (of for example one trading day):
\(\beta_t = \frac{\text{EWMA}(r_{S,t})}{\sqrt{h_r}}\) where
\(r_{S,t} = \log\frac{S_t}{S_{t-h_r}}\).

Given \cref{eqSigmaImplFalseDiff}, the quantity \(\zeta_t^i\) is simply
the vol-of-vol evaluated in
\(\bigl(\log\frac{K_i}{f(T_i-t)S_t},T_i-t,\xi_t\bigr)\) and it could be
also computed as a EWMA volatility on historical absolute returns of the
implied volatility surface at the fixed log-forward moneyness and
time-to-maturity grid point, rescaled by the square-root of \(h_r\). For
liquidity reasons, it is also possible to approximate the latter
quantity as the vol-of-vol at the \(1\text{M ATM}\) point, times an
appropriate factor (see
\cref{coverage-performances-of-the-short-term-model-free-var} for the
description of a possible way to calibrate such a factor). In
\cite{wong2021procyclicality}, the authors suggest to consider the
historical series of the \(1\text{M ATM}\) implied volatility point. A
less procyclical alternative identified by the authors consists in
rescaling the VVIX historical value with a sigmoid function which
ensures a smooth vol-of-vol transition between high and low volatility
regimes.

The correlation parameter \(\rho_t\) can be computed using a EWMA
correlation between spot log-returns and absolute IV returns, where the
IV point considered can be again the \(1\text{M ATM}\) point.

Lastly, the derivative \(\partial_k\sigma_t^i\) is the derivative of the
smile with respect to the log-forward moneyness, evaluated in
\(\bigl(\log\frac{K_i}{f(T_i-t)S_t},T_i-t,\xi_t\bigr)\). Since options
are quoted in strike and maturity rather than log-forward moneyness and
time-to-maturity, observe that
\[\partial_k\sigma_t^i = \partial_k\sigma_t\biggl(\log\frac{K_i}{f(T_i-t)S_t},T_i-t,\xi_t\biggr) = K\partial_K\tilde\sigma^\text{imp}_t(K_i,T_i)\]
where
\(\tilde\sigma^\text{imp}_t(K,T)=\sigma^\text{imp}_t\bigl(\log\frac{K}{f(T-t)S_t},T-t\bigr)\).
The derivative of the strike smile can be recovered by simple
interpolation of market data (for example, using cubic B-splines or
arbitrage-free smile models), or by finite differences of market data.

	\hypertarget{t-student-short-term-model-free-var-formulation}{%
\subsubsection{t-Student short-term model-free VaR
formulation}\label{t-student-short-term-model-free-var-formulation}}

	When calculating risk, a large majority of financial players consider
the distribution of the returns of an underlier \(S_t\) to be t-Student.
The reason is linked to the shape of the probability functions of such
distribution, which are fatter, compared to a classic normal
distribution. In this way the importance of extreme events is higher and
this guarantees a larger conservativeness of the risk model.

	In this section we consider then a t-Student distribution for the
relative returns \(\frac{S_{t+h}-S_t}{S_t}\), with \(\nu_t\) degrees of
freedom, location parameter \(\alpha_th\) and scale parameter
\(\beta_t\sqrt h\). In particular, we consider the model
\begin{equation}\label{eqH1TStudent}
S_{t+h} = S_t(1+\alpha_th + \beta_t T_{t+h})\tag{H1}
\end{equation} where \(T_{t+h}\in\mathbb R\) is a t-Student with
\(\nu_t\) degrees of freedom, null mean and variance equal to \(h\).
Then, the probability density function of \(S_{t+h}\) is
\begin{equation}\label{eqPdfSTstudent}
p_{S_{t+h}}(s) = \frac{\Gamma\bigl(\frac{\nu_t+1}2\bigr)}{\Gamma\bigl(\frac{\nu_t}2\bigr)S_t\beta_t\sqrt{\pi h}} \biggl(1+\frac1{\nu_t}\biggl(\frac{s-S_t(1+\alpha_th)}{S_t\beta_t\sqrt h}\biggr)^2\biggr)^{-\frac{\nu_t+1}2}.
\end{equation}

Also, for every strike \(K\) and maturity \(T\), denoting \(\tau=T-t\),
\(k(s)=\log\frac{K}{f(\tau)s}\), we consider the increment
\(\Delta\sigma_t(k(S_t),\tau):=\sigma_{t+h}(k(S_{t+h}),\tau-h)-\sigma_t(k(S_t),\tau)\)
conditional to \(S_{t+h}\) to be a Gaussian random variable with mean
\(m_t(S_{t+h},k(S_t),\tau)\) and variance
\(\zeta_t(k(S_t),\tau)^2(1-\rho^2_t)h\), with
\[m_t(s,k(S_t),\tau) = \mu_t(k(S_t),\tau)h + \frac{s-S_t(1+\alpha_th)}{\beta_tS_t}\zeta_t(k(S_t),\tau)\rho_t.\]
In other words, we write the conditional implied volatility increments
as \begin{equation}\label{eqH2TStudent}
\Delta\sigma_t(k(S_t),\tau)|(S_{t+h}=s) = m_t(s,k(S_t),\tau) + \zeta_t(k(S_t),\tau)\sqrt{1-\rho^2_t}\sqrt{h} X\tag{H2}
\end{equation} where \(X\) is a standard Gaussian random variable.

\begin{remark}

Since with the above hypothesis we only know the distribution of the implied volatility increments conditional to $S_{t+h}$, it could seem difficult to calibrate parameters $\mu_t$, $\zeta_t$ and $\rho_t$ on market data. However, given a conditional distribution, it is easy to recover the moments of the marginal distribution using the tower property in expectations. Indeed, moments up to the second order of $\Delta\sigma_t(k(S_t),\tau)$ (without the conditioning to $S_{t+h}$) are
\begin{align*}
E[\Delta\sigma_t(k(S_t),\tau)] &= \mu_t(k(S_t),\tau)h\\
\text{Var}[\Delta\sigma_t(k(S_t),\tau)] &= \zeta_t(k(S_t),\tau)^2h\\
\text{Corr}[S_{t+h}, \Delta\sigma_t(k(S_t),\tau)] &= \rho_t.
\end{align*}
This allows to easily calibrate parameters based on the historical mean and variance of $\Delta\sigma_t(k(S_t),\tau)$.

\end{remark}

	In the next paragraphs, we justify the following formula for the
\(h\)-days VaR with confidence level \(\theta\) under
\cref{eqH1TStudent,eqH2TStudent}:
\begin{equation}\label{eqModelFreeVaRTstudent}
\begin{aligned}
\text{VaR}_{\theta,t}(h) &= F_Z^{-1}(1-\theta)\sqrt{c_t^2 + q_t^2 + 2\rho_tc_tq_t}\sqrt h\\
c_t &= \beta_t\bigl(S_t\sum_i\pi_i\Delta_t^i - \sum_i\pi_i\mathcal V_t^i\partial_k\sigma_t^i\bigr)\\
q_t &= \sum_i\pi_i\zeta_t^i\mathcal V_t^i
\end{aligned}\tag{Short-term t-Student}
\end{equation} where \begin{equation}\label{eqZTStudent}
Z = \frac{q_t\sqrt{1-\rho^2_t}X + (c_t+q_t\rho_t)Y}{\sqrt{c_t^2+q_t^2+2\rho_tc_tq_t}}
\end{equation} and \(X\) is a standard Gaussian random variable and
\(Y\) is a standard t-Student with \(\nu_t\) degrees of freedom
independent of \(X\).

	Quantities that enter \cref{eqModelFreeVaRTstudent} are the same as in
the Gaussian case: the Black-Scholes Greeks delta \(\Delta_t^i\) and
vega \(\mathcal V_t^i\), the volatility \(\beta_t\) of the underlying
spot \(S_t\), the vol-of-vol \(\zeta_t^i\), the correlation \(\rho_t\),
and the derivative \(\partial_k\sigma_t^i\) of the smile with respect to
the log-forward moneyness. See \cref{a-short-term-model-free-formula}
for a description of how to compute these quantities in practice from
market data.

	\begin{remark}

\cite{berg2010density} shows in Theorem 1 that the probability density function of $Z$ is
$$p_Z(z) = \sum_{k=0}^\infty\phi_k^{(\nu_t,\gamma)}g_{k,a_1}(z)$$
where $a_1=\frac{q\sqrt{1-\rho_t^2}}{\sqrt{c_t^2+q_t^2+2\rho_tc_tq_t}}$, $\gamma = \frac{c_t+q_t\rho_t}{q_t\sqrt{2(1-\rho_t^2)}}$ and
\begin{align*}
\phi_k^{(\nu_t,\gamma)} &= \frac{\Gamma\bigl(k+\frac{1}2\bigr)}{k!\Gamma\bigl(\frac12\bigr)\Gamma\bigl(\frac{\nu_t}2\bigr)}\int_0^\infty \exp(-f)f^{\frac{\nu_t-1}2}\bigl(f+\gamma^2\bigr)^{-k-\frac12}\,df\\
g_{k,a_1}(z) &= \frac{\Gamma\bigl(\frac{1}2\bigr)}{\Gamma\bigl(k+\frac12\bigr)a_1\sqrt{2\pi}}\biggl(\frac{z^2}{2a_1^2}\biggr)^k\exp\biggl(-\frac{z^2}{2a_1^2}\biggr).
\end{align*}

Quantiles of $Z$ can also be computed empirically, simulating the distribution of the linear combination of two independent random variables distributed as a standard Gaussian and a standard t-Student with $\nu_t$ degrees of freedom. In particular, simulations of $Z$ can be found following the steps:

\begin{enumerate}
\item Simulate two independent normal random variables $X$ and $N_Y$, with mean $0$ and variance $1$;
\item Turn $N_Y$ into a uniform distribution $U_Y=\Phi(N_Y)$;
\item Recover the t-Student random variable via $Y=F_t^{-1}(U_Y;\nu_t)$ where $F_t(\cdot;\nu_t)$ is the cumulative density function of a t-Student with $\nu_t$ degrees of freedom;
\item Put $Z=\frac{q_t\sqrt{1-\rho^2_t}X + (c_t+q_t\rho_t)Y}{\sqrt{c_t^2+q_t^2+2\rho_tc_tq_t}}$.
\end{enumerate}

\end{remark}

	We now explain the rationale of \cref{eqModelFreeVaRTstudent}. Doing a
first order approximation of increments of the option portfolio
\(\Pi_t(S_t,\sigma_t)\), taking into consideration the dependence of
every options' implied volatility to the log-forward moneyness and so to
the underlier, we find that the form of the \(\text{P\&L}\)s is
\begin{equation}\label{eqPnLTStudentModelFree}
\begin{aligned}
\Pi_{t+h}(S_{t+h},\sigma_{t+h})-\Pi_t(S_t,\sigma_t) &\approx \beta_t\bigl(S_t\sum_i\pi_i\partial_S\text{BS}_t^i - \sum_i\pi_i\partial_k\sigma_t^i\partial_\sigma\text{BS}_t^i\bigr)\,T_{t+h} + \sum_i\pi_i\partial_\sigma\text{BS}_t^i\Delta\sigma^i_t\\
&= c_tT_{t+h} + \sum_i\pi_i\mathcal V_t^i\Delta\sigma^i_t,
\end{aligned}
\end{equation} where we used the same notations as in
\cref{a-short-term-model-free-formula}. Here, we do not know the
distribution of the increments of the implied volatilities, so that we
cannot automatically infer the distribution of the \(\text{P\&L}\)s.
However, we can still compute VaRs using the relation in
\cref{eqIntPPnLCondS}.

Firstly, we can write \(T_{t+h} = \sqrt h \tilde Y\) where \(\tilde Y\)
is a standard t-Student with \(\nu_t\) degrees of freedom, and
\[\Delta\sigma_t(k(S_t),\tau) = \mu_t(k(S_t),\tau)h + \zeta_t(k(S_t),\tau)\sqrt h \tilde X\]
for a certain random variable \(\tilde X\) with mean \(0\) and variance
\(1\). Then, given \cref{eqPnLTStudentModelFree}, the distribution of
\(\frac{\text{P\&L}}{\sqrt h}\) tends to the distribution of
\(c_t\tilde Y + \sum_i\pi_i\mathcal V_t^i\zeta_t^i \tilde X\), which
does not depend on \(h\). In particular since \begin{align*}
1-\theta &= P\bigl(\text{P\&L}\leq v(\theta,h)\bigr)\\
&= P\biggl(\frac{\text{P\&L}}{\sqrt h}\leq \frac{v(\theta,h)}{\sqrt h}\biggr),
\end{align*} and the limiting random variable has a strictly positive
density, then the function \(v(\theta,h)\) is asymptotic with
\(\sqrt h\), i.e.~\(v(\theta,h) = u(\theta)\sqrt{h} + o(\sqrt h)\). This
is consistent with \cref{eqModelFreeVaR}, where the \(h\)-days VaR is
proportional to the square-root of \(h\).

	Secondly, the distribution of the \(\text{P\&L}\)s conditional to
\(S_{t+h}\) is Gaussian and in particular
\[P\bigl(\text{P\&L}\leq v(\theta,h)|S_{t+h}=s\bigr) = \Phi\biggl(\frac{v(\theta,h)-c_tt(s)-\sum_i\pi_i\mathcal V_t^im_t^i(s)}{q_t\sqrt{1-\rho^2_t}\sqrt{h}}\biggr)\]
where \(t(s) = \frac{s-S_t(1+\alpha_th)}{\beta_tS_t}\). Then, removing
the conditionality to the probability of the \(\text{P\&L}\)s, it holds
\begin{align*}
P\bigl(\text{P\&L}\leq v(\theta,h)\bigr) &= \int_{-\infty}^\infty p_{S_{t+h}}(s)\Phi\biggl(\frac{v(\theta,h)-c_tt(s)-\sum_i\pi_i\mathcal V_t^im_t^i(s)}{q_t\sqrt{1-\rho^2_t}\sqrt{h}}\biggr)\,ds\\
&= \int_{-\infty}^\infty p_{T}(y)\Phi\biggl(\frac{v(\theta,h)-c_ty\sqrt{h}-\sum_i\pi_i\mathcal V_t^im_t^i\bigl(S_t(1+y\beta_t\sqrt{h}+\alpha_th)\bigr)}{q_t\sqrt{1-\rho^2_t}\sqrt{h}}\biggr)\,dy
\end{align*} where \(p_{S_{t+h}}\) is as in \cref{eqPdfSTstudent},
\(p_{T}\) is the probability density function of a standard t-Student
with \(\nu\) degrees of freedom, and we used the transformation
\(y=\frac{t(s)}{\sqrt{h}}\). In this way, for the Lebesgue's dominated
convergence theorem and using the fact that
\(v(\theta,h) = u(\theta)\sqrt{h} + o(\sqrt h)\), the right hand side of
the previous relation goes to
\[\int_{-\infty}^\infty p_{T}(y)\Phi\biggl(\frac{u(\theta)-(c_t+q_t\rho_t)y}{q_t\sqrt{1-\rho^2_t}}\biggr)\,dy\]
for \(h\) going to \(0\). Consider two independent random variables
\(X\) and \(Y\) with \(X\) a standard Gaussian and \(Y\) a standard
t-Student with \(\nu_t\) degrees of freedom. We can write the latter
expression as
\[E\biggl[P\biggl(X\leq\frac{u(\theta)-(c_t+q_t\rho_t)Y}{q_t\sqrt{1-\rho^2_t}}\biggr)\biggr]=P\biggl(X\leq\frac{u(\theta)-(c_t+q_t\rho_t)Y}{q_t\sqrt{1-\rho^2_t}}\biggr).\]
Defining the random variable \(Z\) as in \cref{eqZTStudent}, we shall
look at the value of \(u(\theta)\) such that
\[1-\theta = P\biggl(Z\leq\frac{u(\theta)}{\sqrt{c_t^2+q_t^2+2\rho_tc_tq_t}}\biggr).\]

	All in all, the short-term model-free VaR formula in the t-Student case
becomes \cref{eqModelFreeVaRTstudent} ignoring terms in \(o(\sqrt h)\).

	\begin{remark}

In both \cref{eqModelFreeVaR} and \cref{eqModelFreeVaRTstudent} the vol-of-vol parameter depends on the strike and maturity of the option, while the correlation does not. This is due to the underlying hypothesis that the whole implied volatility surface is driven by one single Brownian motion, even though the magnitude of movements for each surface point depends on the point itself. The short-term model-free formulas can be generalized to the case where there is more than one Brownian motion driving the implied volatility surface (typically the target dimension is of $2$ or $3$).

\end{remark}

	\hypertarget{properties-and-limitations}{%
\subsection{Properties and
limitations}\label{properties-and-limitations}}

	\hypertarget{local-quantities-and-extreme-risk-concrete-practical-implementation}{%
\subsubsection{Local quantities and extreme risk: concrete practical
implementation}\label{local-quantities-and-extreme-risk-concrete-practical-implementation}}

	Observe that all the above VaR estimations (the Black-Scholes formula,
the Stochastic volatility formula, and the short-term model-free
formula) are defined with local quantities: deltas, vegas, instantaneous
volatility and correlation coefficients. The initial margin however
incorporates a tail risk which looks at future moves in prices that
typically correspond to large moves. Even if there is an apparent
paradox here, the explanation is clear: those formulas are asymptotic
formulas when the time step \(h\) goes to zero, and for sufficient small
\(h\) even the tail risk will be driven by the local quantities, in so
far as we deal we diffusion models.

The whole question therefore is how those asymptotic formulas will
behave in practice. Obviously, the smaller the MPOR, or the less
volatile the market, the better. A careful backtesting will be the clue
here: it will allow to diagnose whether the coverage and procyclicality
behavior of the formula are satisfactory.

In this regard, and especially from a regulatory perspective, one should
keep in mind that the final IM formula will contain other components
besides this core one, like a weighted Stress Historical VaR and the
Short Option Minimum quantity described in
\cref{the-mechanism-of-initial-margin-for-options}. In general the
former component will be obtained by computing price returns along
stress historical scenarios with full re-evaluation (meaning, using the
Black-Scholes formula for options with the shocked underlier and implied
volatility) instead of the local first order Greeks. Therefore the risk
of missing a convexity behavior should be largely mitigated, if not
fully eliminated. Regarding the SOM, consider a portfolio of short deep
OTM options. Today, this portfolio has negligible delta and vega
quantities, so that the VaR estimation is around \(0\), even though
there actually is a tail risk. This hidden risk is far to be local, but
it still should be taken into consideration in the initial margin
calculation. This is the rationale of the SOM, which is already
implemented by CCPs and should cover the risk of those short-term
portfolios, as discussed in
\cref{the-mechanism-of-initial-margin-for-options}.

	\hypertarget{symmetry-with-respect-to-the-portfolio}{%
\subsubsection{Symmetry with respect to the
portfolio}\label{symmetry-with-respect-to-the-portfolio}}

	It is easy to see that all the new VaR formulas in this article are
symmetrical with respect to the portfolio, i.e.~being short or long on
the same portfolio would produce the same VaR exposure. This could seem
weird, especially when we suppose a log-normal distribution of the spot,
which is not symmetric. The symmetricity appears when we take the limit
for \(h\) going to \(0\). Indeed, the terms multiplying \(\sqrt{h}\) are
symmetrical in the portfolio position while the ones that should break
the symmetricity multiply higher orders of \(h\), so that they are
canceled out in the limit.

However, this symmetry is not an issue when computing margins: as seen
in \cref{the-mechanism-of-initial-margin-for-options}, the final total
risk requirement charged by the CCP is composed of the margin computed
on P\&Ls (refined by the add-ons and the SOM) minus the NOV component.
In this way, neglecting the add-ons and the SOM, a long portfolio
\(\Pi>0\) with initial margin \(\text{IM}\) has a total risk requirement
equal to \(\text{IM}-\Pi\); while the same portfolio but on a short
position \(-\Pi<0\) implies a total risk requirement of
\(\text{IM}+\Pi\).

	\hypertarget{comparison-with-fhs}{%
\subsubsection{Comparison with FHS}\label{comparison-with-fhs}}

	In \cref{limitations-of-the-fhs}, we have seen that among its drawbacks,
the FHS model is limited by the number of scenarios that it can
generate, depending on the available historical data. On the other hand,
the short-term model-free formula in
\cref{a-short-term-model-free-formula} does not need to compute
simulated scenarios and eventually requires historical data only for the
calibration of volatility parameters.

	Secondly, while the FHS does not capture the joint dynamics of risk
factors in complex products, the short-term model-free formula in
\cref{a-short-term-model-free-formula} considers both the singular
margin impact of each risk factor and the joint margin impact affected
by the correlation of risk factors. Furthermore, the short-term
model-free formula for options is more natural than the FHS methodology,
whose application to IV surface points is more subtle.

A third limitation is the difficulty of FHS to generate arbitrage-free
scenarios, which is not an issue for the short-term model-free formula
in \cref{a-short-term-model-free-formula} since it does not require the
generation of scenarios and does not face the arbitrage issue.

	Finally, regarding the procyclicality of the VaR estimation, we show in
numerical experiments in
\cref{coverage-performances-of-the-short-term-model-free-var} that the
short-term model-free VaR is less procyclical then the FHS VaR for the
tested portfolios.

	We turn now to the exact computation of the VaR in the neural-SDE model.

	\hypertarget{quasi-explicit-formula-for-the-var-in-the-neural-sde-model}{%
\section{Quasi-explicit formula for the VaR in the neural-SDE
model}\label{quasi-explicit-formula-for-the-var-in-the-neural-sde-model}}

	In this section we investigate the neural-SDE model described in
\cref{the-neural-sde-model} and the special specification of its
parameters with the aim of applying it to an IM computation. We are not
interested in the calibration of arbitrage-free call prices surfaces via
neural networks but to the affine factor model for normalized option
prices itself, so that parameters can be calibrated with any algorithm
of choice, which is not necessarily a neural network.

The model is particularly simple and it turns out to have a
quasi-explicit formula for the VaR of option portfolios, as we show in
\cref{quasi-explicit-formula-for-the-var} below. In practice, this could
enable rapid computations for the IM in such models, which may prove to
be highly relevant when properly calibrated.

Moreover, the VaR can be approximated by a closed formula which is
proportional to the square-root of the MPOR (see
\cref{closed-formula-for-the-short-term-var}). This approximated formula
coincides with the VaR formula in the Stochastic Volatility model of
\cref{the-black-scholes-case-and-the-stochastic-volatility-case}.

	We reemphasize the fact that while the model can be calibrated also in
different ways as the ones described in \cite{cohen2021arbitrage}, the
results in this chapter are still valid and independent from the
calibration setup.

We use the same notations as in
\cref{a-simple-short-term-model-free-formula}. Furthermore, in the whole
section, the notation \(\lVert \cdot\rVert_2\) indicates the Euclidean
\(2\)-norm,
i.e.~\(\lVert (a_1,\dots,a_d)^T\rVert_2^2 = \sum_{i=1}^da_i^2\).

	\hypertarget{quasi-explicit-formula-for-the-var}{%
\subsection{Quasi-explicit formula for the
VaR}\label{quasi-explicit-formula-for-the-var}}

	Consider a portfolio of Vanilla calls with price at time \(t\) given by
\begin{equation}\label{eqPi}
\begin{aligned}
\Pi_t(S_t,\xi_t) &= \sum_i\pi_iC\Bigl(T_i-t,\log\frac{K_i}{f(T_i-t)S_t};S_t,\xi_t\Bigr)\\
&= S_t\sum_i\pi_i\text{DF}_t(T_i-t)f(T_i-t) c\Bigl(T_i-t,\log\frac{K_i}{f(T_i-t)S_t};\xi_t\Bigr)
\end{aligned}
\end{equation} where
\[c\Bigl(T_i-t,\log\frac{K_i}{f(T_i-t)S_t};\xi_t\Bigr)=G_0\Bigl(T_i-t,\log\frac{K_i}{f(T_i-t)S_t}\Bigr)+G\Bigl(T_i-t,\log\frac{K_i}{f(T_i-t)S_t}\Bigr)\cdot\xi_{t}.\]
From now on, we work at time \(t\), so that quantities \(S_t\) and
\(\xi_t\) are known. The
\(\text{P\&L}:=\Pi_{t+h}(S_{t+h},\xi_{t+h})-\Pi_t(S_t,\xi_t)\) of the
portfolio reads then \begin{equation}\label{eqPnL}
\text{P\&L} = A(h,S_{t+h}) + B(h,S_{t+h})\cdot\bigl(\xi_{t+h}-\xi_t\bigr)
\end{equation} where \begin{equation}\label{eqAB}
\begin{aligned}
A(h,s) =&\ s\sum_i\pi_i\text{DF}_t(T_i-(t+h))f(T_i-(t+h))\biggl(G_0\Bigl(T_i-(t+h),\log\frac{K_i}{f(T_i-(t+h))s}\Bigr) +\\
&+ G\Bigl(T_i-(t+h),\log\frac{K_i}{f(T_i-(t+h))s}\Bigr)\cdot\xi_t\biggr) - \Pi_t(S_t,\xi_t),\\
B(h,s) =&\ s\sum_i\pi_i\text{DF}_t(T_i-(t+h))f(T_i-(t+h))G\Bigl(T_i-(t+h),\log\frac{K_i}{f(T_i-(t+h))s}\Bigr).
\end{aligned}
\end{equation}

An important consequence to the representation in \cref{eqPnL} is the
linearity of the P\&Ls in \(\xi_{t+h}-\xi_t\). In terms of VaR
calculations, this means that the VaR for the P\&Ls' distribution
conditional to \(S_{t+h}\) is linear with respect to the VaR for the
\(\xi_{t+h}\) distribution.

	\hypertarget{hypothesis-on-the-joint-increments}{%
\subsubsection{Hypothesis on the joint
increments}\label{hypothesis-on-the-joint-increments}}

	Since from a practical perspective market data is always related to a
discrete time grid, from now on, for risk calculations we consider
processes defined via their Euler scheme as in
\cref{empirical-var-in-the-neural-sde-model}, i.e.
\begin{equation}\label{eqEulerScheme}
\begin{aligned}
S_{t+h} &= S_t\exp\biggl(\Bigl(\alpha_t-\frac{\beta_t^2}2\Bigr) h + \beta_t(W_{0,t+h}-W_{0,t})\biggr),\\
\xi_{t+h} &= \xi_t + \mu_t h + \sigma_t\cdot(W_{t+h}-W_t)
\end{aligned}
\end{equation} where \(W_{0,t+h}-W_{0,t}\in\mathbb R\) and
\(W_{t+h}-W_t\in\mathbb{R}^d\) are Gaussian random variables with
combined law \(N(\mathbf 0, hP_t)\) and \(P_t\) is the correlation
matrix \begin{equation*}
P_t=\begin{pmatrix}
    1 & \begin{matrix} P_{S,\xi,t}^T \end{matrix}\\
    \begin{matrix} P_{S,\xi,t} \end{matrix} & P_{\xi,t}
\end{pmatrix}.
\end{equation*}

We work at time \(t\), so that quantities \(S_t\), \(\xi_t\),
\(\alpha_t\), \(\beta_t\), \(\mu_t\) and \(\sigma_t\) are known. We will
not need to observe \(W_{0,t}\) and \(W_{t}\).

	\begin{remark}\label{remarkSXiSDE}

The Euler schemes with time step $h$ for the processes $S_t$ and $\xi_t$ defined via the SDE \cref{eqSXi} are a particular case of \cref{eqEulerScheme}, therefore the results of this section hold also in this case.

\end{remark}

	To develop \cref{eqIntPPnLCondS} we need a partial result regarding the
distribution of the increments of \(\xi_t\) conditional to \(S_{t+h}\).

	\begin{lemma}\label{lemmaDistribXiCondS}

For processes in \cref{eqEulerScheme}, the distribution of $\xi_{t+h}-\xi_t$ conditional to $S_{t+h}=s$ is a Gaussian $N(m_t(s),V_t)$ where
\begin{align*}
m_t(s) &= \mu_t h + \frac{1}{\beta_t}\biggl(\log\frac{s}{S_t}-\Bigl(\alpha_t-\frac{\beta_t^2}2\Bigr)h\biggr)\sigma_t\cdot P_{S,\xi,t},\\
V_t &= h(\sigma_t\cdot b_t)\cdot(\sigma_t\cdot b_t)^T,
\end{align*}
and $b_t\in\mathbb{R}^{d\times d}$ is a matrix such that $b_t\cdot b_t^T = P_{\xi,t}-P_{S,\xi,t}\cdot P_{S,\xi,t}^T$.

In the case of independent processes, $m_t(s)=\mu_th$ and $V_t=h\sigma_t\cdot\sigma_t^T$.

\end{lemma}

	\begin{proof}

Let us consider the Brownian increments $\Delta W_{0,t} = W_{0,t+h}-W_{0,t}$ and $\Delta W_t = W_{t+h}-W_{t}$, which have Gaussian joint distribution $N(\mathbf{0},hP_t)$. The Gaussian random variable $Z = \Delta W_t - \Delta W_{0,t}P_{S,\xi}$ has null mean and variance equal to $h(P_{\xi,t}-P_{S,\xi,t}\cdot P_{S,\xi,t}^T)$. Since covariance matrices are symmetric and positive semi-definite, there exists a matrix $b_t\in\mathbb{R}^{d\times d}$ such that $b_t\cdot b_t^T = P_{\xi,t}-P_{S,\xi,t}\cdot P_{S,\xi,t}^T$. Also, $Z$ and $\Delta W_{0,t}$ are independent since uncorrelated and jointly Gaussian, and it follows that the distribution of $\Delta W_t$ conditional to $\Delta W_{0,t}=w_0$ is a Gaussian with mean $w_0P_{S,\xi,t}$ and covariance matrix $hb_t\cdot b_t^T$.

From \cref{eqEulerScheme}, it is immediate to recover the distribution of $\xi_{t+h}-\xi_t$ conditional to $\Delta W_{0,t}=w_0$. For the conditionality with respect to $S_{t+h}=s$, it is enough to substitute $w_0$ with $\frac{1}{\beta_t}\bigl(\log\frac{s}{S_t}-\bigl(\alpha_t-\frac{\beta_t^2}2\bigr)h\bigr)$.

If processes are independent, $P_{S,\xi,t}=\mathbf{0}$, $P_{\xi,t}=I_d$, and the conclusion follows.

\end{proof}

	As an immediate consequence to \cref{lemmaDistribXiCondS}, we can write
the increments of \(\xi_t\) conditional to \(S_{t+h}\) as
\begin{equation}\label{eqXiCondS}
\xi_{t+h}-\xi_t|(S_{t+h}=s) = m_t(s) + \sqrt{h}\sigma_t\cdot b_t\cdot X
\end{equation} where \(X\sim N(0,I_d)\) is a Gaussian random variable
independent to \(S_{t+h}\).

Then \begin{align*}
\text{P\&L}|(S_{t+h}=s) &= A(h,s) + B(h,s)\cdot\bigl(m_t(s) + \sqrt{h}\sigma_t\cdot b_t\cdot X\bigr)\\
&= \hat A(h,s) + \hat B(h,s)\cdot X
\end{align*} where \begin{equation}\label{defHatAHatB}
\begin{aligned}
\hat A(h,s) &:= A(h,s) + B(h,s)\cdot m_t(s) \in\mathbb R,\\
\hat B(h,s) &:= B(h,s)\cdot\sqrt{h}\sigma_t\cdot b_t \in\mathbb R^{1\times d}.
\end{aligned}
\end{equation}

	In particular, conditional to \(S_{t+h}=s\), the P\&L is a sum of
jointly Gaussian variables, so it is also a Gaussian variable with law
\(N(\hat A(h,s),\lVert\hat B(h,s)\rVert_2^2)\). Then the quantity
\(P(\text{P\&L}\leq v(\theta,h)|S_{t+h}=s)\) is the cumulative function
of a Gaussian variable, and in particular it is equal to
\begin{equation*}
P(\text{P\&L}\leq v(\theta,h)|S_{t+h}=s) = \Phi\biggl(\frac{v(\theta,h)-\hat A(h,s)}{\lVert\hat B(h,s)\rVert_2}\biggr).
\end{equation*}

Reconsidering \cref{eqIntPPnLCondS}, the VaR at risk level \(\theta\)
for the P\&Ls can be computed as specified in the following lemma.

	\begin{proposition}\label{lemmaXiGauss}

Under the model of \cref{eqCReisModel,eqEulerScheme}, the $h$-days VaR at confidence level $\theta$ of the portfolio \cref{eqPi} is the value of $v(\theta,h)$ which solves
\begin{equation}\label{eqIntfF}
1-\theta = \int_0^\infty \Phi\biggl(\frac{v(\theta,h)-\hat A(h,s)}{\lVert\hat B(h,s)\rVert_2}\biggr)\,dF_{S_{t+h}}(s)s.
\end{equation}
where \cref{eqAB,defHatAHatB} define $\hat A(h,s)$ and $\hat B(h,s)$.

\end{proposition}

	Note that \cref{lemmaXiGauss} gives a semi-closed formula for the VaR in
the neural-SDE model, with no need of further hypothesis. As a
consequence, we can compute efficiently the VaR in this model without
using any approximation.

	We shall notice that since losses cannot be larger than today's
position, the result \(v(\theta,h)\) should always be higher than minus
the current value of the portfolio,
i.e.~\(v(\theta,h)\geq-\Pi_t(S_t,\xi_t)\). This condition holds true if
and only if the P\&Ls' distribution is null below \(-\Pi_t(S_t,\xi_t)\),
or equivalently if and only if the distribution of future prices is null
below \(0\). In particular, it must hold \[G\cdot\xi_{t+h}>-G_0\] for
any \(\xi_{t+h}\). This condition does not seem to be guaranteed a
priori. Indeed, it depends on how the parameters of the distribution of
the \(\xi_t\) are calibrated. However, if the \(\xi_t\) are calibrated
such that call prices always satisfy no arbitrage conditions, then in
particular prices will always be positive.

	\hypertarget{calls-and-puts-portfolio}{%
\subsubsection{Calls and Puts
portfolio}\label{calls-and-puts-portfolio}}

	In the case of portfolios with both call and put options, it is
sufficient to re-write put options using the put-call-parity
\[P\Bigl(T-t,\log\frac{K}{F_t(T-t)}\Bigr) = C\Bigl(T-t,\log\frac{K}{F_t(T-t)}\Bigr) - \text{DF}_t(T-t)\bigl(F_t(T-t)-K\bigr).\]
In this way, the relation in \cref{eqPnL} still holds redefining
quantities \(A\) and \(B\) with elementary steps. In particular for
every put option position
\(\pi P\bigl(T-t,\log\frac{K}{F_t(T-t)}\bigr)\) in the portfolio,
\(A(h,s)\) adds the term \begin{align*}
&s\pi\text{DF}_t(T-(t+h))f(T-(t+h))\biggl(G_0\Bigl(T-(t+h),\log\frac{K}{f(T-(t+h))s}\Bigr) +\\
&+ G\Bigl(T-(t+h),\log\frac{K}{f(T-(t+h))s}\Bigr)\cdot\xi_t - 1\biggr) + \pi\text{DF}_t(T-(t+h))K,
\end{align*} with \(\Pi_t(S_t,\xi_t)\) updating its value with the added
puts, while \(B(h,s)\) adds
\[s\pi\text{DF}_t(T-(t+h))f(T-(t+h))G\Bigl(T-(t+h),\log\frac{K}{f(T-(t+h))s}\Bigr).\]

	\hypertarget{closed-formula-for-the-short-term-var}{%
\subsection{Closed formula for the short term
VaR}\label{closed-formula-for-the-short-term-var}}

	In this section we start by proving that the VaR in the neural-SDE model
for option prices is of the form \(u(\theta)\sqrt{h}\) asymptotically
with \(h\). This formulation reflects empirical results and standard
models adopted in industry and it is consistent with
\cref{eqModelFreeVaR}. Then, we state the main result of this section
computing the explicit form of the function \(u(\theta)\).

	Firstly, we look at the form of the function \(v(\theta,h)\) when \(h\)
is small.

\begin{lemma}\label{lemmaVaRSqrtH}

Under the model of \cref{eqCReisModel,eqEulerScheme}, the $h$-days VaR at confidence level $\theta$ of the portfolio \cref{eqPi} is asymptotic to $\sqrt{h}$ for $h$ going to $0$:
$$\text{VaR}_{\theta,t}(h) = u(\theta)\sqrt{h} + o\bigl(\sqrt h\bigr)$$
for a certain function $u(\theta)$ not depending on $h$.

\end{lemma}

We give the proof in \cref{proof-of-lemma}.

	Our next result uses the quantities \begin{equation}\label{eqDefCDRho}
\begin{aligned}
c_t :=&\ S_t\beta_t \sum_i\pi_i\text{DF}_t(T_i-t)f(T_i-t)\times\\
&\times\Bigl[G_0\Bigl(T_i-t,\log\frac{K_i}{f(T_i-t)S_t}\Bigr) + G\Bigl(T_i-t,\log\frac{K_i}{f(T_i-t)S_t}\Bigr)\cdot\xi_t  - \mathbbm{1}_{\text{Puts}}(i) +\\
&- \partial_k\Bigl(G_0\Bigl(T_i-t,\log\frac{K_i}{f(T_i-t)S_t}\Bigr) + G\Bigl(T_i-t,\log\frac{K_i}{f(T_i-t)S_t}\Bigr)\cdot\xi_t\Bigr)\Bigr] +\\
&+B(0,S_t)\cdot\sigma_t\cdot P_{S,\xi,t},\\
q_t :=&\ \bigl\lVert B(0,S_t)\cdot\sigma_t\cdot b_t \bigr\rVert_2,\\
\end{aligned}
\end{equation} where \(\mathbbm{1}_{\text{Puts}}(i)\) is \(1\) if the
index \(i\) refers to a put, otherwise it is null, and
\[B(0,S_t) = S_t\sum_i\pi_i\text{DF}_t(T_i-t)f(T_i-t)G\Bigl(T_i-t,\log\frac{K_i}{f(T_i-t)S_t}\Bigr).\]

\begin{remark}

It is easy to prove that $c_t$ and $q_t$ can be written with the alternative expressions:
\begin{align*}
c_t &= S_t\beta_t\frac{d}{dS_t}\Pi_t(S_t,\xi_t) + \nabla_{\xi_t}\Pi_t(S_t,\xi_t)^T\cdot\sigma_t\cdot P_{S,\xi,t},\\
q_t &= \lVert\nabla_{\xi_t}\Pi_t(S_t,\xi_t)^T\cdot\sigma_t\cdot b_t\rVert_2.\\
\end{align*}

\end{remark}

	We can now state the main result of this section.

\begin{proposition}\label{propFinal}

Under the model of \cref{eqCReisModel,eqEulerScheme}, the $h$-days VaR at confidence level $\theta$ of the portfolio \cref{eqPi} is
\begin{equation}\label{eqVTheta}
\text{VaR}_{\theta,t}(h) = \Phi^{-1}\bigl(1-\theta\bigr)\sqrt{c_t^2+q_t^2}\sqrt{h} + o\bigl(\sqrt{h}\bigr)
\end{equation}
where $c_t$ and $q_t$ are defined in \cref{eqDefCDRho}.

\end{proposition}

	The proof is given in \cref{proof-of}.

	\begin{corollary}\label{corolPropFinal}

Under the model of \cref{eqCReisModel,eqEulerScheme} with $d=1$, the $h$-days VaR at confidence level $\theta$ of the portfolio \cref{eqPi} is
\begin{equation*}
\begin{aligned}
\text{VaR}_{\theta,t}(h) =&\ \Phi^{-1}\bigl(1-\theta\bigr)\times\\
&\times \sqrt{\Bigl(S_t\beta_t\frac{d}{dS_t}\Pi_t(S_t,\xi_t)\Bigr)^2 + \Bigl(\sigma_t\frac{d}{d\xi_t}\Pi_t(S_t,\xi_t)\Bigr)^2 + 2P_{S,\xi,t}S_t\beta_t\sigma_t\frac{d}{dS_t}\Pi_t(S_t,\xi_t)\frac{d}{d\xi_t}\Pi_t(S_t,\xi_t)}\sqrt{h} +\\
&+ o\bigl(\sqrt h\bigr).
\end{aligned}
\end{equation*}

\end{corollary}

	Observe that this is compatible with the Stochastic Volatility model's
VaR in \cref{eqGenericVaR}.

	\hypertarget{normal-distribution-for-s_th}{%
\subsubsection{\texorpdfstring{Normal distribution for
\(S_{t+h}\)}{Normal distribution for S\_\{t+h\}}}\label{normal-distribution-for-s_th}}

	\Cref{lemmaDistribXiCondS} still holds considering normal increments for
\(S_{t+h}\):
\[S_{t+h} = S_t\Bigl(1+\alpha_th + \beta_t(W_{0,t+h}-W_{0,t})\Bigr),\]
appropriately redefining the quantity \(m_t(s)\). Indeed, in this case
\(m_t(s)\) becomes
\(\mu_th + \frac{s - S_t(1 + \alpha_th)}{S_t\beta_t}\sigma_t\cdot P_{S,\xi,t}\).

The results in \cref{lemmaVaRSqrtH,propFinal} are exactly the same as
for the log-normal case. Indeed, the probability density function of
\(S_{t+h}\) is
\[p_{S_{t+h}}(s) = \frac{1}{S_t\beta_t\sqrt{2\pi h}}\exp\Biggl(-\Biggl(\frac{s - S_t(1+\alpha_th)}{S_t\beta_t\sqrt{2h}}\Biggr)^2\Biggr)\]
and the conclusion can be easily attained as in the previous case, using
the transformation \(y=\frac{s-S_t(1+\alpha_th)}{S_t\beta_t\sqrt{h}}\)
as done in \cref{proof-of}.

	\hypertarget{t-student-distribution-for-s_th}{%
\subsubsection{\texorpdfstring{t-Student distribution for
\(S_{t+h}\)}{t-Student distribution for S\_\{t+h\}}}\label{t-student-distribution-for-s_th}}

	As in \cref{t-student-short-term-model-free-var-formulation}, we
consider the case where the increments of the underlier \(S_t\) follow a
t-Student distribution, while increments of the implied volatility
\(\sigma_t\) are Gaussian conditional to \(S_{t+h}\).

\Cref{lemmaDistribXiCondS} holds true in the Gaussian case because the
random variable \(Z = \Delta W_t - \Delta W_{0,t}P_{S,\xi}\) is still
Gaussian and its decorrelation with \(\Delta W_{0,t}\) implies its
independence. However, in this case \(\Delta W_{0,t}\) is substituted
with the t-Student \(T_t\) and we cannot derive the same result. We need
then some additional hypothesis to derive the equivalent of
\cref{lemmaXiGauss} when \(S_{t+h}\) is a t-Student, and in particular
we shall take \cref{eqXiCondS} as granted a priori.

	\begin{lemma}\label{lemmaXiGaussTStudent}

Consider the model of \cref{eqCReisModel} and the hypothesis that
$$S_{t+h} = S_t(1+\alpha_th + \beta_t T_{t+h})$$
where $T_{t+h}\in\mathbb R$ is a t-Student with $\nu_t$ degrees of freedom, null mean and variance equal to $h$, and $\xi_{t+h}-\xi_t$ conditional to $S_{t+h}=s$ is a Gaussian random variable with mean $m_t(s)$ and covariance matrix $V_t$ with
\begin{align*}
m_t(s) &= \mu_th + \frac{s - S_t(1 + \alpha_th)}{S_t\beta_t}\sigma_t\cdot P_{S,\xi,t},\\
V_t &= h(\sigma_t\cdot b_t)\cdot(\sigma_t\cdot b_t)^T,
\end{align*}
and $b_t\in\mathbb{R}^{d\times d}$ is a matrix such that $b_t\cdot b_t^T = P_{\xi,t}-P_{S,\xi,t}\cdot P_{S,\xi,t}^T$, for certain parameters $\mu_t$, $P_{S,\xi,t}\in\mathbb{R}^d$, $\sigma_t$, $P_{\xi,t}\in\mathbb{R}^{d\times d}$. Then the $h$-days VaR at confidence level $\theta$ of the portfolio \cref{eqPi} is the value of $v(\theta,h)$ which solves
\begin{equation*}
1-\theta = \int_0^\infty \Phi\biggl(\frac{v(\theta,h)-\hat A(h,s)}{\lVert\hat B(h,s)\rVert_2}\biggr)\,dF_{S_{t+h}}(s)s.
\end{equation*}
where \cref{eqAB,defHatAHatB} define $\hat A(h,s)$ and $\hat B(h,s)$.

\end{lemma}

	The calibration of the parameters \(\mu_t\), \(P_{S,\xi,t}\),
\(\sigma_t\) and \(P_{\xi,t}\) is practically difficult if performed
from the distribution of \(\xi_{t+h}-\xi_t\) conditional to \(S_{t+h}\).
However, as in \cref{t-student-short-term-model-free-var-formulation},
we can recover the moments of the marginal distribution, allowing an
easy calibration of the parameters based on the historical mean and
covariances of \(\xi_{t+h}-\xi_t\).

	\begin{corollary}\label{corollDistribXiNoCond}

Under the model
$$S_{t+h} = S_t(1+\alpha_th + \beta_t T_{t+h})$$
where $T_{t+h}\in\mathbb R$ is a t-Student with $\nu_t$ degrees of freedom, null mean and variance equal to $h$, if $\xi_{t+h}-\xi_t$ conditional to $S_{t+h}=s$ is a Gaussian random variable with mean $m_t(s)$ and covariance matrix $V_t$ as in \cref{lemmaXiGaussTStudent}, then
\begin{align*}
E[\xi_{t+h}-\xi_t] &= \mu_th\\
\text{Cov}[\xi_{t+h}-\xi_t] &= h\sigma_t\cdot P_{\xi,t}\cdot\sigma_t^T\\
\text{Cov}\biggl[\begin{pmatrix}S_{t+h} \\ \xi_{t+h}-\xi_t\end{pmatrix}\biggr] &= \begin{pmatrix}
    \beta_t^2S_t^2h & \begin{matrix} \beta_tS_th(\sigma_t\cdot P_{S,\xi,t})^T \end{matrix}\\
    \begin{matrix} \beta_tS_th\sigma_t\cdot P_{S,\xi,t} \end{matrix} & h\sigma_t\cdot P_{\xi,t}\cdot\sigma_t^T
\end{pmatrix}.
\end{align*}

\end{corollary}

	The result in \cref{lemmaVaRSqrtH} still holds true in the t-Student
case. The proof is equivalent to the one given in \cref{proof-of-lemma},
with the adaptations regarding the distribution of \(S_{t+h}\) and
\(\xi_{t+h}-\xi_t\). These adaptations can be found in the proof of
\cref{eqModelFreeVaRTstudent} in
\cref{t-student-short-term-model-free-var-formulation}. As a
consequence, the function \(v(\theta,h)\) is of the form
\(u(\theta)\sqrt h+o(\sqrt h)\) for \(h\) small.

	The probability density function of \(S_{t+h}\) is as in
\cref{eqPdfSTstudent}. In this way, using the transformation
\(y=\frac{s-S_t(1+\alpha_th)}{S_t\beta_t\sqrt{h}}\) and repeating the
calculations in \cref{proof-of-the-pointwise-convergence}, we find
\cref{eqIntPhiH0} with exactly the same values for \(c_t\) and \(q_t\).
We then look for the value of \(u(\theta)\) such that
\[1-\theta = E\biggl[P\biggl(\frac{q_tX+c_tY}{\sqrt{c_t^2+q_t^2}}\leq \frac{u(\theta)}{\sqrt{c_t^2+q_t^2}}\biggr)\biggr].\]

	The random variable \(Z=\frac{q_tX+c_tY}{\sqrt{c_t^2+q_t^2}}\) is not
Gaussian in this case since it is the sum of a standard Gaussian random
variable and a standard t-Student random variable, which are
independent. However, it still holds that the initial margin in the case
of t-Student returns is \begin{equation*}
\text{VaR}_{\theta,t}(h) = F_Z^{-1}\bigl(1-\theta\bigr)\sqrt{c_t^2+q_t^2}\sqrt{h} + o\bigl(\sqrt h\bigr).
\end{equation*}

The empirical quantile of \(Z\) can be recovered as detailed in
\cref{t-student-short-term-model-free-var-formulation}.

	The above observations lead us to the following.

\begin{corollary}\label{corollTStudent}

Under the framework of \cref{lemmaXiGaussTStudent}, the $h$-days VaR at confidence level $\theta$ of the portfolio \cref{eqPi} is
\begin{equation*}
\text{VaR}_{\theta,t}(h) = F^{-1}_Z\bigl(1-\theta\bigr)\sqrt{c_t^2+q_t^2}\sqrt{h} + o\bigl(\sqrt h\bigr)
\end{equation*}
where $c_t$, $q_t$ are defined in \cref{eqDefCDRho} and $Z=\frac{qX+cY}{\sqrt{c_t^2+q_t^2}}$ with $X$ and $Y$ independent and $X$ standard Gaussian, $Y$ standard t-Student with $\nu_t$ degrees of freedom.

\end{corollary}

	\hypertarget{numerical-experiments}{%
\section{Numerical experiments}\label{numerical-experiments}}

	\hypertarget{backtesting-option-portfolios}{%
\subsection{Backtesting option
portfolios}\label{backtesting-option-portfolios}}

	Before going into the numerical experiments, it is worth focusing on the
specific issues arising while backtesting option portfolios. Indeed,
options are contracts with a fixed strike and a fixed expiry, so that
considering a backtest on a real fixed contract is awkward for two
reasons: at maturity the option will expire and the backtest could not
continue anymore; the option could become far OTM or ITM in time,
completing losing interest in the market and becoming illiquid or even
not traded anymore.

Moreover, in order to focus on a given risk (like the calendar spread
one), and its adequate coverage by the margin model, it is better to
consider a portfolio with a constant risk profile across time, and so
constant specifications in terms of moneyness and time-to-maturity.

For this reason, options are generally backtested for fixed log-forward
moneyness (or delta) and fixed time-to-maturity, rather than fixed
contract. Of course, these desired options are not always available
among the market quoted ones, so that two possibilities arise:

\begin{enumerate}
\def\labelenumi{\arabic{enumi}.}
\tightlist
\item
  Considering the nearest in log-forward moneyness (or delta) and
  time-to-maturity real quoted option.
\item
  Considering synthetic option prices on the chosen fixed log-forward
  moneyness (or delta) and fixed time-to-maturity obtained via the model
  pricing criteria.
\end{enumerate}

In the first case, the backtested portfolios will possibly change every
day, depending on how much the ATM level has moved and on the rolling
maturity. In the second case, the VaR estimations are compared to model
\(\text{P\&L}\)s rather than real ones, so that if the calibration model
is not good enough, backtesting results could be misleading.

In general, there is no preferred way to backtest option portfolios and
CCPs may adopt both methodologies. We recommend to perform the two
approaches for production, because the synthetic option prices, due to
the complexity of the data treatments performed, may eventually not
reflect fully faithfully the effective market returns when they are
available, as used directly in the first approach above.

	\hypertarget{var-formula-in-the-heston-model}{%
\subsection{VaR formula in the Heston
model}\label{var-formula-in-the-heston-model}}

	The VaR formula in \cref{eqGenericVaR} can be applied to any Stochastic
Volatility model. In this section we test it on the classic Heston model
\begin{align*}
dS_t &= \alpha_tS_tdt + \sqrt{\nu_t}S_tdW_{0,t}\\
d\nu_t &= \kappa(\theta-\nu_t)dt + \xi\sqrt{\nu_t}dW_t\\
dW_{0,t}dW_t &= \rho dt.
\end{align*} In this case, \cref{eqGenericVaR} has the form
\begin{equation}\label{eqHestonVaR}
\text{VaR}_{\theta,t}(h) = \Phi^{-1}(1-\theta)\sqrt{S_t^2\nu_t\Bigl(\frac{d}{dS_t}\Pi_t(S_t,\nu_t)\Bigr)^2 + \xi^2\nu_t\Bigl(\frac{d}{d\nu_t}\Pi_t(S_t,\nu_t)\Bigr)^2 + 2\rho \xi S_t\nu_t\frac{d}{dS_t}\Pi_t(S_t,\beta_t)\frac{d}{d\nu_t}\Pi_t(S_t,\nu_t)}\sqrt h.
\end{equation}

	Firstly, we simulate one year (\(365\) days) history of the process
\((S_t,\nu_t)\) with a simple Euler scheme of the form \begin{align*}
S_{t+\delta t} &= S_t\bigl(1 + \alpha_t\delta t + \sqrt{\nu_t}\sqrt{\delta t}X_0\bigr)\\
\nu_{t+\delta t} &= \nu_t + \kappa(\theta-\nu_t)\delta t + \xi\sqrt{\nu_t}\sqrt{\delta t}X
\end{align*} where \(X_0\), \(X\) are standard Gaussian random variables
with correlation \(\rho\). Differences are negligible using the
log-formulation for \(S_t\), i.e.
\[S_{t+\delta t} = S_t\exp\Bigl(\Bigl(\alpha_t-\frac{\nu_t}2\Bigr)\delta t + \sqrt{\nu_t}\sqrt{\delta t}X_0\Bigr),\]
or using a Milstein scheme instead of the Euler's one.

At this point, for different outright, calendar and butterfly
portfolios, we compute daily prices using the semi-analytical formula
for a call option in the Heston model described in \cite{kahl2005not}.
We take null rates, so that the forward price coincides with the spot
value and the discount factor is \(1\). Finally we compare real PnLs
with \(0.99\)-VaR estimations as in \cref{eqHestonVaR} on different MPOR
horizons. In order to compute portfolio derivatives with respect to the
spot and the volatility of the spot, we use the average between the
corresponding backward and forward finite differences: \begin{align*}
\frac{d}{dS_t}\Pi_t(S_t,\nu_t) &\approx \frac{\Pi_t(S_t+\varepsilon,\nu_t)-\Pi_t(S_t-\varepsilon,\nu_t)}{2\varepsilon},\\
\frac{d}{d\nu_t}\Pi_t(S_t,\nu_t) &\approx \frac{\Pi_t(S_t,\nu_t+\varepsilon)-\Pi_t(S_t,\nu_t-\varepsilon)}{2\varepsilon}.
\end{align*}

Note in particular that these quantities are \emph{not} the
Black-Scholes ones defined through the sensitivities of the
Black-Scholes formula evaluated at the implied volatility corresponding
to the Heston model price.

	We use calibrated Heston parameters on S\&P500 on December 2015 (see
\cite{cao2021pricing}), in particular
\((\kappa, \sqrt\theta, \xi, \rho)=(6.169, 0.16168, 0.477, -0.781)\),
and we start with initial values \(S_0 = 2054\) and
\(\sqrt{\nu_0} = 0.15562\). For the Euler scheme, we choose a step
\(\delta t\) of \(10^{-1}\) days. The portfolios considered are outright
calls with delta in \(\{0.2, 0.35, \delta_{\text{ATM}}, 0.65, 0.8\}\)
and time-to-maturity in \(\{30, 90, 180, 365\}\) days, and the resulting
combinations of calendar spread and butterfly spread portfolios. In
particular, a calendar spread is a portfolio with one short call at a
fixed strike \(K\) and maturity \(T_1\) and one long call with same
strike \(K\) and maturity \(T_2>T_1\). For a fixed delta, the common
strike of the two options is chosen to be the one related to the
shortest maturity. A butterfly spread is composed of two long calls with
deltas \(\delta\) and \(1-\delta\) respectively, and two short ATM
calls. For butterfly spreads, we also test the deltas
\(0.1, 0.3, 0.4, 0.45\).

The number of tested portfolios is then: \(5\times4 = 20\) outrights,
\({4\choose 2}\times5 = 30\) calendar spreads, and \(6\times4 = 24\)
butterfly spreads; in total \(74\) portfolios.

For each tested portfolio, we compute the coverage ratio as the number
of days where the model VaR covers the realized loss over the total
number of tested days, and the size of losses as the ratio between the
margin loss (difference between realized loss and model VaR) and the
portfolio price. The average and the median over all portfolios is
displayed in \cref{tableCoverageHeston}. The results are very
satisfactory for all MPORs. Indeed, the average coverage meets the
\(0.99\)-VaR requirement and breaches are of a very small size below
\(7\%\) of the portfolio value.

\begin{table}
    \centering
    \begin{tabular}{ |c|c|c||c|c| }
    \hline
    MPOR (days) & \multicolumn{2}{|c|}{Coverage} & \multicolumn{2}{|c|}{Size of loss}\\
    \hline
     & Mean & Median & Mean & Median\\
    \hline
    $1$ & $0.9927$ & $0.9945$ & $0.0485$ & $0.0204$\\
    $2$ & $0.9902$ & $0.9945$ & $0.0645$ & $0.0353$\\
    $3$ & $0.9896$ & $0.9945$ & $0.0682$ & $0.0422$\\
    \hline
    \end{tabular}
    \caption{Average coverage and size of loss of \cref{eqHestonVaR} on $74$ option portfolios on simulated Heston data.}
    \label{tableCoverageHeston}
\end{table}

	Similarly, we compute the initial margin using the short-term model-free
formula described in \cref{a-short-term-model-free-formula}. In
particular, we simulate \(5\) years history of an Heston process with
same parameters as in the previous test and compute the initial margin
for the same portfolios on the last year's observations (the previous
history is used to calibrate parameters). The formula used is then
\cref{eqModelFreeVaR}, where the delta and the vega Greeks are the
Black-Scholes ones, the spot volatility, the vol-of-vol and the
correlation between risk factors are computed with the EWMA
specification, and the derivative of the implied volatility with respect
to the log-forward moneyness in computed via finite differences. The
implied volatility point used to compute the EWMA correlation is the
\(1\text{M ATM}\) point.

As in the previous test, we compute the average coverage and size of
loss for each tested MPOR. Results are shown in
\cref{tableCoverageHeston2}. Results are less conservative than in the
previous test since the coverage is around \(0.983\). However, the size
of loss is still very small compared to the portfolio value, and
actually smaller than in the previous test.

\begin{table}
    \centering
    \begin{tabular}{ |c|c|c||c|c| }
    \hline
    MPOR (days) & \multicolumn{2}{|c|}{Coverage} & \multicolumn{2}{|c|}{Size of loss}\\
    \hline
     & Mean & Median & Mean & Median\\
    \hline
    $1$ & $0.9853$ & $0.9863$ & $0.0404$ & $0.0194$\\
    $2$ & $0.9826$ & $0.9822$ & $0.0565$ & $0.0368$\\
    $3$ & $0.9813$ & $0.9808$ & $0.0496$ & $0.0255$\\
    \hline
    \end{tabular}
    \caption{Average coverage and size of loss of the short-term model-free VaR \cref{eqModelFreeVaR} on $74$ option portfolios on simulated Heston data, under the hypothesis of log-normal distribution of spot returns.}
    \label{tableCoverageHeston2}
\end{table}

Results can actually be improved using the hypothesis of a t-Student
distribution for spot returns and redefining the VaR for MPOR \(h\) as
in \cref{eqModelFreeVaRTstudent} where the distribution of \(Z\) is
obtained empirically as explained in
\cref{t-student-short-term-model-free-var-formulation}.
\Cref{tableCoverageHeston3} shows results when considering \(5\) degrees
of freedom in the t-Student distribution of spot returns. As expected,
results are more conservative than the normal case and satisfy the
\(0.99\) coverage requirement.

\begin{table}
    \centering
    \begin{tabular}{ |c|c|c||c|c| }
    \hline
    MPOR (days) & \multicolumn{2}{|c|}{Coverage} & \multicolumn{2}{|c|}{Size of loss}\\
    \hline
     & Mean & Median & Mean & Median\\
    \hline
    $1$ & $0.9907$ & $0.9918$ & $0.0362$ & $0.0164$\\
    $2$ & $0.9886$ & $0.9918$ & $0.0596$ & $0.0364$\\
    $3$ & $0.9889$ & $0.9945$ & $0.0485$ & $0.0291$\\
    \hline
    \end{tabular}
    \caption{Average coverage and size of loss of the short-term model-free VaR \cref{eqModelFreeVaRTstudent} on $74$ option portfolios on simulated Heston data, under the hypothesis of t-Student with $5$ degrees of freedom distribution of spot returns.}
    \label{tableCoverageHeston3}
\end{table}

	\hypertarget{coverage-performances-of-the-short-term-model-free-var}{%
\subsection{Coverage performances of the short-term model-free
VaR}\label{coverage-performances-of-the-short-term-model-free-var}}

	In this section we show the results of coverage of the short-term
model-free \(0.99\) VaR formula in
\cref{a-short-term-model-free-formula} compared with the classical FHS
model described in \cref{filtered-historical-simulation}.

We use a database of S\&P500 data provided to Zeliade by the Clearify
project\footnote{Clearify is a collaboration between Zeliade and the Imperial College Mathematical Finance Department funded by an Imperial Faculty of Natural Sciences Strategic Research Funding Award.}
on end of the day option prices. Firstly, we clean the rough data
removing options with \(0\) volume and use the put-call-parity on mid
prices to extrapolate forward and discount factors for each quoted
maturity having at least two put-call couples. Then, we remove all calls
not satisfying the arbitrage bounds
\[\text{DF}_t(T)(F_t(T)-K)^+\leq C_t(T,K)\leq \text{DF}_t(T)F_t(T),\]
and all puts not satisfying
\[\text{DF}_t(T)(K-F_t(T))^+\leq P_t(T,K)\leq \text{DF}_t(T)K.\] At this
points, puts are transformed into calls and all the following
computations are performed for call prices.

In order to get normalized historical prices, we compute synthetic
historical prices on a fixed grid of time-to-maturity and log-forward
moneyness. At this aim, we observe that market data is generally dense
around the ATM point for short maturities and it spreads out with
increasing expiry. For this reason, the most cunning fixed grid should
be in delta, so we identify \(17\) delta points from \(0.015\) to
\(0.985\) and compute the corresponding log-forward moneyness for a
symbolic volatility of \(0.1\), over the grid of time-to-maturities of
\(2\), \(5\), \(10\), \(21\), \(42\), \(63\), \(126\) and \(252\) days.
Observe that in this way the grid is not constant in log-forward
moneyness for different time-to-maturities. At this point, we firstly
compute implied volatilities based on real prices. Then, we interpolate
linearly on the log-forward moneyness direction since data is dense
enough. The interpolation on the time-to-maturity direction is done on
the implied total variances (squared implied volatilities times the
time-to-maturity) adding synthetic points for the zero maturity equal to
\(0\) and then interpolating linearly. The interpolated prices are then
corrected to avoid static arbitrage as described in
\cite{cohen2020detecting}.

For comparison, we implemented a second interpolation scheme firstly
adding synthetic points for the zero maturity (setting prices equal to
their intrinsic values) and for extreme moneyness (with prices equal to
the discounted forward on the left and null prices on the right); then
normalizing all prices by their discounted forward; finally using
monotonic cubic splines on the log-forward moneyness direction and
linear interpolation on the time-to-maturity direction. The final
results reported in this section do not significantly change.

	We consider two portfolios: the first one is an ATM calendar spread
between maturities of \(1\)M and \(6\)M; while the second one is a
butterfly spread on maturity \(3\)M and moneyness \(0.9\), \(1\),
\(1.1\). We compute the VaR for an MPOR horizon of \(1\) day and a
confidence level \(\theta=0.99\).

	The short-term model-free VaR is obtained computing:

\begin{enumerate}
\item The spot volatility $\beta_t$ via a EWMA volatility algorithm with decay factor $0.97$ on spot log-returns, divided by the square-root of the daily step;
\item The correlation $\rho_t$ via a EWMA correlation on spot log-returns and ATM implied volatility absolute returns;
\item The vol-of-vol $\zeta_t(k,\tau)$ as $F$ times the $\text{1M ATM}$ vol-of-vol obtained as a EWMA volatility on implied volatility absolute returns, divided by the square root of the daily step. In order to be conservative enough, the $F$ factor is $1.1$ times the quantile $0.9$ of $5$ years history of ratios between the $(k,\tau)$ vol-of-vol and the $\text{1M ATM}$ vol-of-vol;
\item Delta and vega quantities as the Black-Scholes deltas and vegas evaluated at the option implied volatility;
\item The derivative of the smile with respect to the log-forward moneyness as the derivative of the interpolated smile via B-splines.
\end{enumerate}

Since we consider log-normal spot returns, we use the VaR formulation
with normal quantiles of \cref{eqModelFreeVaR}.

	The FHS risk factors are the spot prices and the \(17\times 8\) implied
volatility grid points for fixed log-forward moneyness and
time-to-maturity. We consider log-returns for the former risk factors
and absolute returns for the latter ones. The volatility of risk
factors' returns is computed via EWMA with decay factor \(0.97\). The
future discount factors and forward values are obtained under the
assumption of constant risk-free rates in the MPOR horizon.

	We backtest the short-term model-free VaR and the FHS VaR against the
synthetic P\&Ls as explained in the second approach of
\cref{backtesting-option-portfolios}.

	\Cref{figureFHSNM} shows the VaR patterns for the tested portfolios. We
can see that the short-term model-free VaR has more breaches then the
FHS VaR, however these are of small size and can be entirely removed
setting a vol-of-vol factor \(F=1\). Alternatively, one could consider
the t-Student framework described in
\cref{t-student-short-term-model-free-var-formulation}. The most
noticeable feature of the short-term model-free VaR is its regularity
compared to the FHS VaR. In particular, the short-term model-free VaR
behaves as we expect after large moves in realized P\&Ls, and it also
softens its behavior, without big jumps. On the contrary, the FHS VaR is
not as consistent (and sometimes seems to move without following market
patterns).

Furthermore, the short-term model-free VaR looks more smooth, i.e.~less
procyclical. To prove this sentence, we compute the peak-to-trough ratio
on the whole \(2019\) dates and the average \(n\)-day procyclicality
measure (in percentage) for \(n\) equal to \(1\), \(5\), \(10\) and
\(20\) days. In particular, the two quantities are computed as
\begin{align*}
\text{Peak-to-trough} &= \frac{\max_t\bigl(-\text{VaR}_{0.99,t}(h)\bigr)}{\min_t\bigl(-\text{VaR}_{0.99,t}(h)\bigr)}\\
\text{$n$-day $\%$} &= \text{max}_t\biggl(\frac{-\text{VaR}_{0.99,t}(h)}{-\text{VaR}_{0.99,t-n}(h)}-1\biggr)\times 100
\end{align*} where \(t\) ranges in the whole \(2019\) and we choose an
MPOR \(h=1\). Results are displayed in \cref{tableFHSNM}. We see that
except for the \(1\)-day procyclicality measure in the calendar spread
portfolio, all other procyclicality measures for both portfolios are
largely smaller for the short-term model-free VaR, i.e.~the latter model
is less procyclical than the FHS VaR.

\begin{figure}
	\hspace*{-1cm}
	\begin{subfigure}{.7\textwidth}
		\includegraphics[width=.8\linewidth]{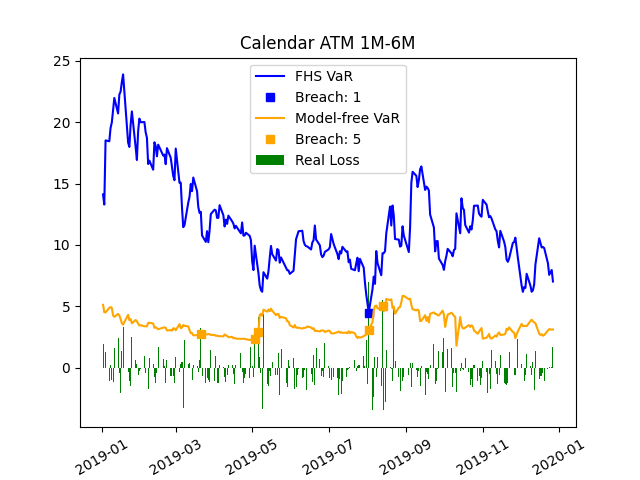}
	\end{subfigure}\hspace{-2.8cm}%
	\begin{subfigure}{.7\textwidth}
		\includegraphics[width=.8\linewidth]{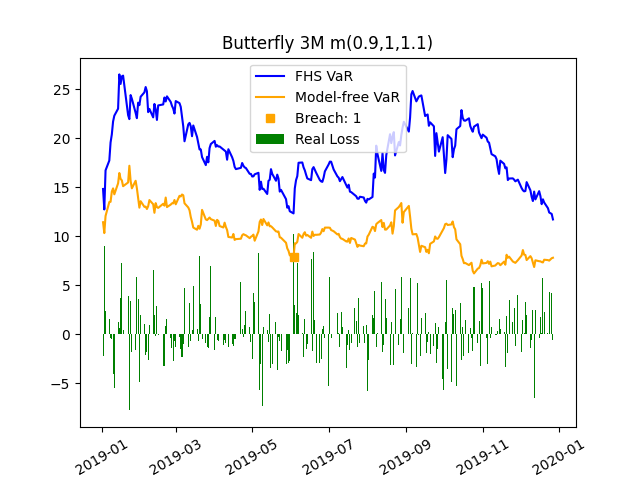}
	\end{subfigure}
	\caption{Margins obtained with the FHS algorithm (blue) and the short-term model-free VaR (orange), for a calendar spread ATM $1$M-$6$M portfolio (left) and a butterfly spread $3$M with moneyness $(0.9,1,1.1)$ portfolio (right).}
	\label{figureFHSNM}
\end{figure}

	\begin{table}
    \centering
    \begin{tabular}{ |c|c|c||c|c| }
    \hline
    & \multicolumn{2}{|c|}{Calendar ATM $1$M-$6$M} & \multicolumn{2}{|c|}{Butterfly $3$M m$(0.9,1,1.1)$}\\
    \hline
    & FHS & Short-term model-free & FHS & Short-term model-free\\
    \hline  
    \text{Peak-to-trough} & $5.32$ & $3.27$ & $2.26$ & $2.77$\\
    \text{$1$-day $\%$} & $42.17$ & $45.38$ & $31.34$ & $17.95$\\
    \text{$5$-day $\%$} & $90.11$ & $89.04$ & $70.47$ & $44.23$\\
    \text{$10$-day $\%$} & $143.88$ & $109.28$ & $106.69$ & $52.17$\\
    \text{$20$-day $\%$} & $139.75$ & $118.53$ & $90.35$ & $44.56$\\
    \hline
    \end{tabular}
    \caption{Comparison between FHS VaR and short-term model-free VaR peak-to-trough ratio and average percentage $n$-day procyclicality measure for a calendar spread and a butterfly spread portfolios.}
    \label{tableFHSNM}
\end{table}

	\hypertarget{practical-implementation-of-the-neural-sde-model}{%
\subsection{Practical implementation of the neural-SDE
model}\label{practical-implementation-of-the-neural-sde-model}}

	In this section we consider the neural-SDE model for normalized option
prices, and in particular we compare VaR estimations obtained as
empirical quantiles on simulations (see
\cref{empirical-var-in-the-neural-sde-model}) and VaR values resulting
from the approximated closed formula of \cref{propFinal}.

	We use the same data as in
\cref{coverage-performances-of-the-short-term-model-free-var} and
compute forward and discount factors similarly.

Before calibrating the \(G\) factors on the time-to-maturity,
log-forward moneyness grid identified in
\cref{coverage-performances-of-the-short-term-model-free-var},
historical prices must be interpolated on such a grid. With this aim, we
use the same interpolation/extrapolation algorithm consisting in implied
volatility's linear interpolation on the space dimension and total
variances' linear interpolation in the time direction.

	Now, for a fixed date \(t\), the past \(5\) years historical data
(\(1260\) observations) is used to calibrate the \(G\) and \(\xi\)
factors. We choose to calibrate the factors \(\xi_s\) for \(s\leq t\) in
the most efficient way, only looking at the statistical accuracy. In
particular, for the fixed time-to-maturity and log-forward moneyness
grid, we choose \(G_0\) as the average historical prices and the
remaining \(G_i\) as the principal components of the residuals between
prices and values of \(G_0\). See \cite{cohen2021arbitrage} for a
detailed description of the calibration algorithm. The calibration code
that we use is the one implemented in the Github repository of the cited
article.\footnote{\url{https://github.com/vicaws/neuralSDE-marketmodel}}
Based on calibration accuracy and process time, we decide to take \(2\)
statistical accuracy factors \(\xi_s\).

At this point we have the constant factors \(G\) on the fixed grid and
an history of factors \(\xi_s\). Furthermore, the calibration uses
neural networks to estimate the distributions of \(S\) and \(\xi\) as in
\cref{eqSXi}. In particular, we have today's parameters \(\alpha_{t}\),
\(\beta_{t}\), \(\mu_{t}\) and \(\sigma_{t}\), and for any value of
\(S\) and \(\xi\), the neural network can predict the corresponding
parameters. Observe that we take the covariance matrix \(P_t=I_d\) as in
\cite{cohen2021arbitrage}.

	For future computations, the matrix \(G\) has to be interpolated outside
the fixed time-to-maturity and log-forward moneyness grid. To do so, we
firstly interpolate normalized call prices on the target couple
\((\tau,k)\) for every historical past day. The interpolation is
performed as in the preparation of the initial database, computing
implied volatilities and interpolating them linearly on the space
direction and linearly in total variance on the time direction. Once all
history on \((\tau,k)\) is retrieved, the \(G(\tau,k)\) factors are the
intercept and the coefficients of the linear regression of prices along
the history of \(\xi\)s.

	At present, the two VaR calculation methodologies can be implemented. We
test the same two portfolios of
\cref{coverage-performances-of-the-short-term-model-free-var} consisting
of an ATM calendar spread \(1\)M-\(6\)M, and a butterfly spread on
maturity \(3\)M and moneyness \(0.9\), \(1\), \(1.1\). We consider an
MPOR of \(1\) day and a VaR confidence level \(\theta=0.99\).

For both VaR methodologies, we work under the assumption of constant
risk-free rates in the MPOR horizon. Then, the discount factor
\(\text{DF}_{t+h}(\tau)\) in \(h\) days on a time-to-maturity \(\tau\)
is equal to today's discount factor \(\text{DF}_{t}(\tau)\), while the
forward value \(F_{t+h}(\tau)\) in \(h\) days on a time-to-maturity
\(\tau\) becomes \(\frac{S_{t+h}}{S_{t}}F_{t}(\tau)\).

For the empirical VaR, simulations are performed under the hypothesis
that parameters \(\alpha\), \(\beta\), \(\mu\) and \(\sigma\) are
constant between \(t\) and \(t+h\), following Euler's scheme in
\cref{eqEulerScheme}. Starting with the estimation of today's
parameters, values of \(S_{t+h}\) and \(\xi_{t+h}\) are simulated
\(10000\) times. Future normalized prices are computed using the model
relation \cref{eqCReisModel} with the \(G\) factors evaluated on
time-to-maturity \(\tau-h\) and log-forward moneyness
\(k+\log\frac{F_t(\tau)}{F_{t+h}(\tau-h)}\), and the estimated values of
\(\xi_{t+h}\). Once the simulated normalized call prices are computed,
they are re-denormalized multiplying by \(\text{DF}_{t}(\tau-h)\) and
\(F_{t+h}(\tau-h)\). The final VaR is the \(1-\theta\) empirical
quantile of P\&Ls obtained as difference of simulated future prices and
today price.

In the case of VaR obtained via approximated closed formula, in order to
be consistent with the empirical VaR, the distribution of \(S_{t+h}\) is
taken to be a log-normal distribution, so that the used closed formula
coincides with \cref{eqVTheta}. Derivatives of the components of \(G\)
with respect to \(k\) are computed as the average between backward and
forward finite differences.

	We plot the percentage ratio between the absolute difference between the
two VaR estimations and the portfolio value along year \(2019\) in
\cref{figureReisinger}. We can see that the empirical VaR and the
approximated formula \cref{eqVTheta} for the VaR generally have a very
small error (about \(4\%\) for the calendar portfolio and \(1\%\) for
the butterfly portfolio), with some higher picks which could reach the
\(10\%\) of the portfolio. This is due to the fact that the approximated
formula is less procyclical than the empirical one and reacts slower to
market changes. All in all, the results confirm the consistency of
hypothesis in \cref{propFinal}.

\begin{figure}
	\hspace*{-1cm}
	\begin{subfigure}{.7\textwidth}
		\includegraphics[width=.8\linewidth]{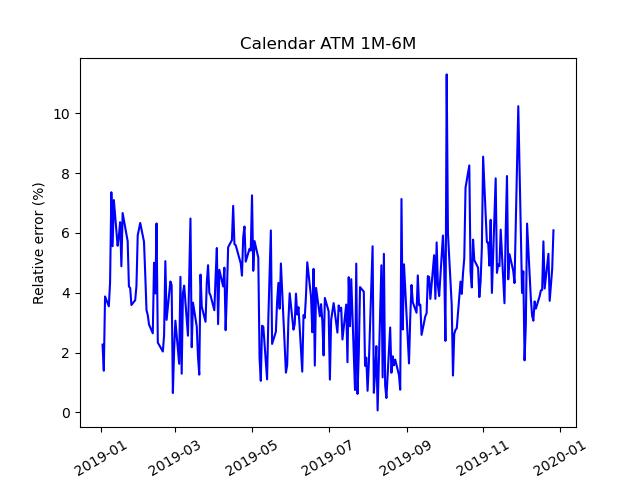}
	\end{subfigure}\hspace{-2.8cm}%
	\begin{subfigure}{.7\textwidth}
		\includegraphics[width=.8\linewidth]{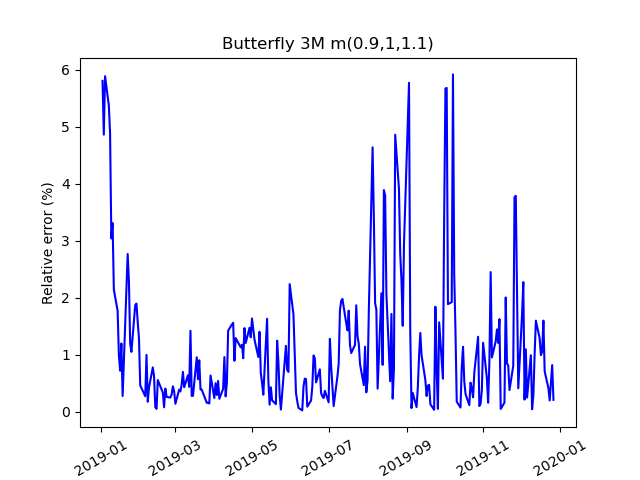}
	\end{subfigure}
	\caption{Percentage relative error between the empirical neural-SDE VaR calculated as in \cref{empirical-var-in-the-neural-sde-model} and the approximated closed formula VaR in \cref{propFinal}, for a calendar spread ATM $1$M-$6$M portfolio (left) and a butterfly spread $3$M with moneyness $(0.9,1,1.1)$ portfolio (right).}
	\label{figureReisinger}
\end{figure}

	\hypertarget{conclusion}{%
\section{Conclusion}\label{conclusion}}

	We summarize and analyze the methodologies that CCPs currently use for
the initial margin of option portfolios. In particular, we compute a
quasi-explicit formula for the VaR of option portfolios in the
neural-SDE model of \cite{cohen2021arbitrage}, and propose a closed
asymptotic short-term model-free formula for the VaR at small time
horizons.

Based on the numerical experiments that we conduct, we are confident
that this new short-term model-free formula could be considered as a
candidate for the core component of IM methodologies for option
portfolios, duly complemented by Short Option Minimum and Stress
Historical VaR components.

	\pagebreak
\appendix

	\hypertarget{proof-of-lemma}{%
\section{\texorpdfstring{Proof of lemma
\ref{lemmaVaRSqrtH}}{Proof of lemma }}\label{proof-of-lemma}}

	Let us consider the distribution of \(\frac{\text{P\&L}}{\sqrt h}\)
given in \cref{eqPnL,eqEulerScheme}. We write
\(W_{0,t+h}-W_{0,t}=\sqrt{h}Y\) and \(W_{t+h}-W_{t}=\sqrt h X\) where
\(Y\) is a standard Gaussian random variable and \(X\) is a
\(d\)-dimensional Gaussian random variable with correlation matrix
\(P_t\) not depending on \(h\).

\begin{normalsize}
We first consider the term $\frac{A(h,S_{t+h})}{\sqrt{h}}$ and look at its limit for $h$ going to $0$. Since $S_{t+h}$ goes to $S_t$, both the numerator and the denominator go to $0$. We then use L'H\^opital's rule to develop the limit. The derivative of $S_{t+h}$ with respect to $h$ is
\end{normalsize}

\[S_{t+h}\biggl(\alpha_t-\frac{\beta_t^2}2 + \frac{\beta_t}{2\sqrt h}Y\biggr)\]
so that the derivative of \(A(h,S_{t+h})\) with respect to \(h\) is
\begin{align*}
&S_{t+h}\Biggl[\Biggl(\alpha_t-\frac{\beta_t^2}2+\frac{\beta_t}{2\sqrt{h}}Y\Biggr)\sum_i\pi_i\text{DF}_t(T_i-(t+h))f(T_i-(t+h))\biggl(G_0\Bigl(T_i-(t+h),\log\frac{K_i}{f(T_i-(t+h))S_{t+h}}\Bigr) +\\
&+G\Bigl(T_i-(t+h),\log\frac{K_i}{f(T_i-(t+h))S_{t+h}}\Bigr)\cdot\xi_t - \mathbbm{1}_{\text{Puts}}(i)\biggr) +\\
+&\sum_i\pi_i\frac{d}{dh}\Bigl(\text{DF}_t(T_i-(t+h))f(T_i-(t+h))\Bigr)\biggl(G_0\Bigl(T_i-(t+h),\log\frac{K_i}{f(T_i-(t+h))S_{t+h}}\Bigr) +\\
&+ G\Bigl(T_i-(t+h),\log\frac{K_i}{f(T_i-(t+h))S_{t+h}}\Bigr)\cdot\xi_t - \mathbbm{1}_{\text{Puts}}(i)\biggr) +\\
-& \sum_i\pi_i\text{DF}_t(T_i-(t+h))f(T_i-(t+h))\biggl(\partial_\tau G_0\Bigl(T_i-(t+h),\log\frac{K_i}{f(T_i-(t+h))S_{t+h}}\Bigr) +\\
&+\partial_\tau G\Bigl(T_i-(t+h),\log\frac{K_i}{f(T_i-(t+h))S_{t+h}}\Bigr)\cdot\xi_t\biggr) +\\
-& \sum_i\pi_i\text{DF}_t(T_i-(t+h))f(T_i-(t+h))\biggl(\partial_k G_0\Bigl(T_i-(t+h),\log\frac{K_i}{f(T_i-(t+h))S_{t+h}}\Bigr) +\\
&+\partial_k G\Bigl(T_i-(t+h),\log\frac{K_i}{f(T_i-(t+h))S_{t+h}}\Bigr)\cdot\xi_t\biggr)\times\\
&\times\Biggl(-\frac{\partial_\tau f(T_i-(t+h))}{f(T_i-(t+h))} + \alpha_t - \frac{\beta_t^2}2 + \frac{\beta_t}{2\sqrt h}Y\Biggr)\Biggr]
\end{align*} where \(\mathbbm{1}_{\text{Puts}}(i)\) is \(1\) if the
index \(i\) refers to a put, otherwise it is null. This quantity
explodes for \(h\) going to \(0\) with speed
\(\frac{\gamma Y}{2\sqrt h}\) where \begin{align*}
\gamma &= S_t\beta_t \sum_i\pi_i\text{DF}_t(T_i-t)f(T_i-t)\times\\
&\times\Bigl[G_0\Bigl(T_i-t,\log\frac{K_i}{f(T_i-t)S_t}\Bigr) + G\Bigl(T_i-t,\log\frac{K_i}{f(T_i-t)S_t}\Bigr)\cdot\xi_t  - \mathbbm{1}_{\text{Puts}}(i) +\\
&- \partial_k\Bigl(G_0\Bigl(T_i-t,\log\frac{K_i}{f(T_i-t)S_t}\Bigr) + G\Bigl(T_i-t,\log\frac{K_i}{f(T_i-t)S_t}\Bigr)\cdot\xi_t\Bigr)\Bigr].
\end{align*}

This means that the ratio \(\frac{A(h,S_{t+h})}{\sqrt{h}}\) tends to
\(\gamma Y\), where \(\gamma\) does not depend on \(h\).

	We now look at the term
\(B(h,S_{t+h})\cdot\frac{\xi_{t+h}-\xi_t}{\sqrt h}\). Firstly, the limit
of \(B(h,S_{t+h})\) for \(h\) going to \(0\) is
\[S_t\sum_i\pi_i\text{DF}_t(T_i-t)f(T_i-t)G\Bigl(T_i-t,\log\frac{K_i}{f(T_i-t)S_t}\Bigr)\]
which is simply a sum of call surfaces and, in general, is different
from \(0\). Secondly, the ratio \(\frac{\xi_{t+h}-\xi_t}{\sqrt h}\) is
equal to \(\mu_t \sqrt{h} + \sigma_t\cdot X\) and goes to
\(\sigma_t\cdot X\) when \(h\) tends to \(0\).

	All in all, \(\frac{\text{P\&L}}{\sqrt h}\) tends to
\(\gamma Y + B(0,S_t)\cdot\sigma_t\cdot X\) almost surely, hence in law.
Then, since the cumulative density function of the random variable
\(\gamma Y + B(0,S_t)\cdot\sigma_t\cdot X\) has a continuous inverse,
all the quantiles of \(\frac{\text{P\&L}}{\sqrt h}\) converge to the
corresponding quantiles of \(\gamma Y + B(0,S_t)\cdot\sigma_t\cdot X\)
as \(h\) tends to \(0\).

The \(h\)-days VaR with confidence level \(\theta\) is the value of
\(v(\theta,h)\) solving \begin{align*}
1-\theta &= P(\text{P\&L}\leq v(\theta,h))\\
&= P\biggl(\frac{\text{P\&L}}{\sqrt h}\leq \frac{v(\theta,h)}{\sqrt h}\biggr).
\end{align*} From the above discussion, we have
\(v(\theta,h) = u(\theta)\sqrt{h} + o(\sqrt h)\), where
\(u(\theta)=F_{\gamma Y + B(0,S_t)\cdot\sigma_t\cdot X}^{-1}(1-\theta)\).

	\hypertarget{proof-of}{%
\section{\texorpdfstring{Proof of
\cref{propFinal}}{Proof of }}\label{proof-of}}

	From the definition in \cref{eqEulerScheme}, \(S_{t+h}\) has a
log-normal distribution with density
\[p_{S_{t+h}}(s) = \frac{1}{s\beta_t\sqrt{2\pi h}}\exp\Biggl(-\Biggl(\frac{\log\frac{s}{S_t} - \bigl(\alpha_t-\frac{\beta_t^2}2\bigr) h}{\beta_t\sqrt{2h}}\Biggr)^2\Biggr).\]
We look at the RHS of \cref{eqIntfF} when \(h\) goes to \(0\).

We can re-write the integral using a change of variable
\(y=\frac{\log\frac{s}{S_t} - \bigl(\alpha_t-\frac{\beta_t^2}2\bigr) h}{\beta_t\sqrt{h}}\)
as
\[\int_{-\infty}^\infty \frac{1}{\sqrt{2\pi}}\exp\biggl(-\frac{y^2}2\biggr)\Phi\biggl(\frac{v(\theta,h)-\hat A(h,s(h,y))}{\lVert\hat B(h,s(h,y))\rVert_2}\biggr)\,dy\]
where
\(s(h,y)=S_t\exp\bigl(y\beta_t\sqrt h + \bigl(\alpha_t-\frac{\beta_t^2}2\bigr) h\bigr)\).
Observe that the integrand is dominated by the integrable function
\(\frac{1}{\sqrt{2\pi}}\exp\bigl(-\frac{y^2}2\bigr)\).

If we prove that for \(h\) going to \(0\) the integrand converges
pointwise to a function of the form
\(\Phi\bigl(\frac{-c_ty + u(\theta)}{q_t}\bigr)\) where \(c_t\) and
\(q_t\) do not depend on \(y\), for the Lebesgue's dominated convergence
theorem the whole integral converges to
\begin{equation}\label{eqIntPhiH0}
\int_{-\infty}^\infty \frac{1}{\sqrt{2\pi}}\exp\biggl(-\frac{y^2}2\biggr)\Phi\biggl(\frac{-c_ty + u(\theta)}{q_t}\biggr)\,dy.
\end{equation}

	Assuming the proof of the convergence is done (see
\cref{proof-of-the-pointwise-convergence}), we can pick-up a pair of
independent standard Gaussian random variables \(X, Y\) and write the
latter expression as
\[E\biggl[\mathbbm 1\biggl(X\leq\frac{-c_tY + u(\theta)}{q_t}\biggr)\biggr].\]
The random variable \(Z=\frac{q_tX+c_tY}{\sqrt{c_t^2 +q_t^2}}\) has a
standard normal distribution, so that the latter quantity is equal to
\[E\biggl[\Phi\biggl(\frac{u(\theta)}{\sqrt{c_t^2 + q_t^2}}\biggr)\biggr] = \Phi\biggl(\frac{u(\theta)}{\sqrt{c_t^2+q_t^2}}\biggr).\]

Then, we can finally recover the expression of the initial margin
\(\text{VaR}_{\theta,t}(h)\) as in \cref{eqVTheta}.

	\hypertarget{proof-of-the-pointwise-convergence}{%
\subsection{Proof of the pointwise
convergence}\label{proof-of-the-pointwise-convergence}}

	The pointwise convergence of
\[\Phi\biggl(\frac{v(\theta,h)-\hat A(h,s(h,y))}{\lVert\hat B(h,s(h,y))\rVert_2}\biggr)\]
to a function of the form
\(\Phi\bigl(\frac{-c_ty + u(\theta)}{q_t}\bigr)\) can be proved firstly
observing that we can work under hypothesis of continuous functions,
given that normalized prices can be considered to be continuous in
time-to-maturity and log-forward moneyness. Also, since we are looking
at the limit when \(h\) is small, we can use \cref{lemmaVaRSqrtH} and
substitute \(v(\theta,h)\) with \(u(\theta)\sqrt h + o(\sqrt{h})\).

Firstly observe that the matrix \(b_t\) does not depend on \(h\), and
\(\lVert \hat B(h,s(h,y))\rVert_2 = \sqrt{h}\lVert B(h,s(h,y))\cdot\sigma_t\cdot b_t\rVert_2\)
where
\[B(h,s(h,y)) = s(h,y)\sum_i\pi_i\text{DF}_t(T_i-(t+h))f(T_i-(t+h))G\Bigl(T_i-(t+h),\log\frac{K_i}{f(T_i-(t+h))s(h,y)}\Bigr).\]
As shown in \cref{proof-of-lemma}, the limit of \(B(h,s(h,y))\) for
\(h\) going to \(0\) is
\[S_t\sum_i\pi_i\text{DF}_t(T_i-t)f(T_i-t)G\Bigl(T_i-t,\log\frac{K_i}{f(T_i-t)S_t}\Bigr)\]
which is different from \(0\). Then, the ratio
\(\frac{u(\theta)\sqrt h + o(\sqrt{h})}{\lVert\hat B(h,s(h,y))\rVert_2}\)
goes to
\[\frac{u(\theta)}{\lVert B(0,S_t)\cdot\sigma_t\cdot b_t\rVert_2}=\frac{u(\theta)}{q_t}\]
in \(0\).

	We now consider the ratio
\(\frac{\hat A(h,s(h,y))}{\sqrt{h}\lVert B(h,s(h,y))\cdot\sigma_t\cdot b_t\rVert_2}\).
The function \(\hat A(h,s(h,y))\) is in turn the sum between
\(A(h,s(h,y))\) and \(B(h,s(h,y))\cdot m_t(s(h,y))\). The latter term is
equal to
\(B(h,s(h,y))\cdot(\mu_th + y\sigma_t\cdot P_{S,\xi,t}\sqrt{h})\), so
that its ratio with
\(\sqrt{h}\lVert B(h,s(h,y))\cdot\sigma_t\cdot b_t\rVert_2\) goes to
\(\frac{B(0,S_t)\cdot\sigma_t\cdot P_{S,\xi,t}}{q_t}y\).

\begin{normalsize}
We shall now focus on the ratio $\frac{A(h,s(h,y))}{\sqrt{h}\lVert B(h,s(h,y))\cdot\sigma_t\cdot b_t\rVert_2}$. Since both the numerator and denominator go to $0$ with $h$, we use L'H\^opital's rule to develop the limit.
\end{normalsize}

The derivative of \(B(h,s(h,y))\) with respect to \(h\) is
\begin{align*}
&s(h,y)\Biggl[\Biggl(\frac{y\beta_t}{2\sqrt{h}} + \alpha_t-\frac{\beta_t^2}2\Biggr)\sum_i\pi_i\text{DF}_t(T_i-(t+h))f(T_i-(t+h))G\Bigl(T_i-(t+h),\log\frac{K_i}{f(T_i-(t+h))s(h,y)}\Bigr) +\\
+& \sum_i\pi_i\frac{d}{dh}\Bigl(\text{DF}_t(T_i-(t+h))f(T_i-(t+h))\Bigr)G\Bigl(T_i-(t+h),\log\frac{K_i}{f(T_i-(t+h))s(h,y)}\Bigr) +\\
-& \sum_i\pi_i\text{DF}_t(T_i-(t+h))f(T_i-(t+h))\partial_\tau G\Bigl(T_i-(t+h),\log\frac{K_i}{f(T_i-(t+h))s(h,y)}\Bigr) +\\
-& \sum_i\pi_i\text{DF}_t(T_i-(t+h))f(T_i-(t+h))\partial_k G\Bigl(T_i-(t+h),\log\frac{K_i}{f(T_i-(t+h))s(h,y)}\Bigr)\times\\
\times&\Biggl(-\frac{\partial_\tau f(T_i-(t+h))}{f(T_i-(t+h))} + \frac{y\beta_t}{2\sqrt h} + \alpha_t - \frac{\beta_t^2}2\Biggr)\Biggr]
\end{align*} and for \(h\) going to \(0\), it explodes with a speed of
\[\frac{1}{2\sqrt h}S_ty\beta_t \sum_i\pi_i\text{DF}_t(T_i-t)f(T_i-t)\Bigl(G\Bigl(T_i-t,\log\frac{K_i}{f(T_i-t)S_t}\Bigr) - \partial_kG\Bigl(T_i-t,\log\frac{K_i}{f(T_i-t)S_t}\Bigr)\Bigr).\]

	Similarly, the derivative of \(A(h,s(h,y))\) with respect to \(h\)
explodes with a speed of \begin{align*}
&\frac{1}{2\sqrt h}S_ty\beta_t \sum_i\pi_i\text{DF}_t(T_i-t)f(T_i-t)\times\\
&\times\Bigl[G_0\Bigl(T_i-t,\log\frac{K_i}{f(T_i-t)S_t}\Bigr) + G\Bigl(T_i-t,\log\frac{K_i}{f(T_i-t)S_t}\Bigr)\cdot\xi_t  - \mathbbm{1}_{\text{Puts}}(i) +\\
&- \partial_k\Bigl(G_0\Bigl(T_i-t,\log\frac{K_i}{f(T_i-t)S_t}\Bigr) + G\Bigl(T_i-t,\log\frac{K_i}{f(T_i-t)S_t}\Bigr)\cdot\xi_t\Bigr)\Bigr].
\end{align*}

Doing the derivative of
\(\sqrt{h}\lVert B(h,s(h,y))\cdot\sigma_t\cdot b_t\rVert_2\), we find
\[\frac{1}{2\sqrt h}\lVert B(h,s(h,y))\cdot\sigma_t\cdot b_t\rVert_2 + \sqrt{h}\frac{\bigl(\frac{d}{dh}B(h,s(h,y))\cdot\sigma_t\cdot b_t\bigr)^T\cdot\bigl(B(h,s(h,y))\cdot\sigma_t\cdot b_t\bigr)}{\lVert B(h,s(h,y))\cdot\sigma_t\cdot b_t\rVert_2},\]
and this explodes with the first term.

\begin{normalsize}
All in all, L'H\^opital's rule shows that the limit of $\frac{\hat A(h,s(h,y))}{\sqrt{h}\lVert B(h,s(h,y))\cdot\sigma_t\cdot b_t\rVert_2}$ for $h$ going to $0$ is $\frac{c_t}{q_t}y$.
\end{normalsize}

    % Add a bibliography block to the postdoc

\newpage \bibliography{Biblio}

\begin{thebibliography}{10}

\bibitem{barone2001non}
Giovanni Barone-Adesi and Kostas Giannopoulos.
\newblock {Non parametric VaR techniques. Myths and realities}.
\newblock {\em Economic Notes}, 30(2):167--181, 2001.

\bibitem{barone1997var}
Giovanni Barone-Adesi, Kostas Giannopoulos, and Les Vosper.
\newblock Var without correlations for nonlinear portfolios.
\newblock {\em Journal of Futures Markets}, 1997.

\bibitem{berg2010density}
Christian Berg and Christophe Vignat.
\newblock {On the density of the sum of two independent Student t-random
  vectors}.
\newblock {\em Statistics \& probability letters}, 80(13-14):1043--1055, 2010.

\bibitem{bergeron2022variational}
Maxime Bergeron, Nicholas Fung, John Hull, Zissis Poulos, and Andreas Veneris.
\newblock {Variational autoencoders: A hands-off approach to volatility}.
\newblock {\em The Journal of Financial Data Science}, 4(2):125--138, 2022.

\bibitem{cao2021pricing}
Jiling Cao, Jeong-Hoon Kim, and Wenjun Zhang.
\newblock {Pricing variance swaps under hybrid CEV and stochastic volatility}.
\newblock {\em Journal of Computational and Applied Mathematics}, 386:113220,
  2021.

\bibitem{cohen2020detecting}
Samuel~N Cohen, Christoph Reisinger, and Sheng Wang.
\newblock {Detecting and repairing arbitrage in traded option prices}.
\newblock {\em Applied Mathematical Finance}, 27(5):345--373, 2020.

\bibitem{cohen2021arbitrage}
Samuel~N Cohen, Christoph Reisinger, and Sheng Wang.
\newblock {Arbitrage-free neural-SDE market models}.
\newblock {\em arXiv preprint arXiv:2105.11053}, 2021.

\bibitem{cohen2022estimating}
Samuel~N Cohen, Christoph Reisinger, and Sheng Wang.
\newblock {Estimating risks of option books using neural-SDE market models}.
\newblock {\em arXiv preprint arXiv:2202.07148}, 2022.

\bibitem{EMIR}
European Commission.
\newblock {Regulation (EU) 648/2012 of 4 July 2012 of the European Parliament
  and Council on OTC derivatives, central counterparties and trade
  repositories}.
\newblock {\em Official Journal of the European Union}, 2012.

\bibitem{RTS_EMIR}
European Commission.
\newblock {Commission Delegated Regulation (EU) No 153/2013 of 19 December
  2012, supplementing Regulation (EU) No 48/2012 of the European Parliament and
  of the Council with regard to regulatory technical standards on requirements
  for central counterparties}.
\newblock {\em Official Journal of the European Union}, 2013.

\bibitem{cont2022simulation}
Rama Cont and Milena Vuleti{\'c}.
\newblock Simulation of arbitrage-free implied volatility surfaces.
\newblock {\em Available at SSRN}, 2022.

\bibitem{glasserman2021w}
Paul Glasserman and Dan Pirjol.
\newblock {W-shaped implied volatility curves and the Gaussian mixture model}.
\newblock {\em Available at SSRN 3951426}, 2021.

\bibitem{gunnarsson2019filtered}
Fredrik Gunnarsson.
\newblock {Filtered Historical Simulation Value at Risk for Options: A
  Dimension Reduction Approach to Model the Volatility Surface Shifts}, 2019.

\bibitem{gurrola2015filtered}
Pedro Gurrola-Perez and David Murphy.
\newblock {Filtered historical simulation Value-at-Risk models and their
  competitors}.
\newblock 2015.

\bibitem{kahl2005not}
Christian Kahl and Peter J{\"a}ckel.
\newblock {Not-so-complex logarithms in the Heston model}.
\newblock {\em Wilmott magazine}, 19(9):94--103, 2005.

\bibitem{ning2021arbitrage}
Brian Ning, Sebastian Jaimungal, Xiaorong Zhang, and Maxime Bergeron.
\newblock {Arbitrage-free implied volatility surface generation with
  variational autoencoders}.
\newblock {\em arXiv preprint arXiv:2108.04941}, 2021.

\bibitem{wong2021procyclicality}
Lauren~W Wong and Yang Zhang.
\newblock {Procyclicality control in risk-based margin models}.
\newblock {\em Journal of Risk}, 23(5), 2021.

\bibitem{yao2017managing}
Jinglun Yao, Sabine Laurent, and Brice B{\'e}naben.
\newblock {Managing Volatility Risk: An Application of Karhunen-Lo{\`{e}}ve
  Decomposition and Filtered Historical Simulation}.
\newblock {\em arXiv preprint arXiv:1710.00859}, 2017.

\end{thebibliography}
\bibliographystyle{plain}

\end{document}